\documentclass[10pt,journal,compsoc]{IEEEtran}

\usepackage{microtype}
\usepackage{graphicx}
\usepackage{subfigure}
\usepackage{booktabs} %
\usepackage{enumitem}
\usepackage{tabularx}
\usepackage{arydshln}
\usepackage{ragged2e}

\usepackage{hyperref}

\usepackage{amsmath}
\usepackage{amssymb}
\usepackage{mathtools}
\usepackage{amsthm}

\usepackage{amsfonts}
\usepackage{nicefrac}

\newcommand{\mI}[0]{\textnormal{I}}
\newcommand{\mJ}[0]{{J}}
\newcommand{\mA}[0]{\textnormal{A}}
\newcommand{\mP}[0]{{P}}
\newcommand{\mcO}[0]{\mathcal{O}}
\newcommand{\dist}[0]{\mathrm{d}}

\newcommand{\BR}{{\mathbb{R}}}
\newcommand{\BE}{{\mathbb{E}}}

\newcommand{\diag}{\mathrm{diag}}

\newcommand{\data}[0]{y}

\newcommand{\bsigma}[0]{\boldsymbol{\sigma}}

\newcommand{\net}[0]{x(\params)}
\newcommand{\netmap}[0]{x(\hat \params)}
\newcommand{\op}[0]{\mA}
\newcommand{\params}[0]{\theta}
\newcommand{\TV}[0]{{\rm TV}}

\providecommand{\argmin}{\operatorname*{argmin}}

\newcommand{\spaced}[1]{\quad \text{#1} \quad}

\usepackage{graphicx}
\usepackage{graphbox}
\usepackage{float}
\usepackage{wrapfig}

\usepackage{cleveref} %

\theoremstyle{plain}
\newtheorem{theorem}{Theorem}[section]
\newtheorem{proposition}[theorem]{Proposition}
\newtheorem{lemma}[theorem]{Lemma}

\theoremstyle{definition}

\theoremstyle{remark}

\usepackage[textsize=tiny]{todonotes}
\ifCLASSOPTIONcompsoc
  \usepackage[nocompress]{cite}
\else
  \usepackage{cite}
\fi
\ifCLASSINFOpdf
\else
\fi
\begin{document}
\title{Uncertainty Estimation for Computed Tomography with a Linearised Deep Image Prior}

\author{Javier~Antorán$^\dagger$,~Riccardo~Barbano$^\dagger$,~Johannes~Leuschner,~José~Miguel~Hernández-Lobato,~Bangti~Jin%
\IEEEcompsocitemizethanks{\IEEEcompsocthanksitem JA and JMHL are from the Department of Engineering, University of Cambridge, United Kingdom.\protect\\
E-mail: ja666@cam.ac.uk
\IEEEcompsocthanksitem RB is from the Department of Computer Science, University College London, United Kingdom\protect\\
E-mail: riccardo.barbano.19@ucl.ac.uk
\IEEEcompsocthanksitem JL is from the Center for Industrial Mathematics, University of Bremen, Germany.\protect
\IEEEcompsocthanksitem BJ is from the Department of Mathematics, The Chinese University of Hong Kong, Shatin, N. T., Hong Kong, P. R. China.\protect\\
$^\dagger$ equal contribution.\protect
}%
}

\IEEEtitleabstractindextext{%
\begin{abstract}
\justifying%
Existing deep-learning based tomographic image reconstruction methods do not provide accurate estimates of reconstruction uncertainty, hindering their real-world deployment.
This paper develops a method, termed as the linearised deep image prior (DIP), to estimate the uncertainty associated with reconstructions produced by the DIP with total variation regularisation (TV).
Specifically, we endow the DIP with conjugate Gaussian-linear model type error-bars computed from a local linearisation of the neural network around its optimised parameters. To preserve conjugacy, we approximate the TV regulariser with a Gaussian surrogate.
This approach provides pixel-wise uncertainty estimates and a marginal likelihood objective for hyperparameter optimisation.
We demonstrate the method on synthetic data and real-measured high-resolution 2D $\mu$CT data, and show that it provides superior calibration of uncertainty estimates relative to previous probabilistic formulations of the~DIP.
Our code is available at \url{https://github.com/educating-dip/bayes_dip}.
\end{abstract}

}

\maketitle

\IEEEdisplaynontitleabstractindextext
\IEEEpeerreviewmaketitle

\IEEEraisesectionheading{\section{Introduction}\label{sec:introduction}}

\IEEEPARstart{I}{}nverse problems in imaging aim to recover
an unknown image $x \in\BR^{d_{x}}$ from the noisy measurement $\data \in \BR^{d_{y}}$
\begin{equation}\label{eq:inverse_problem}
    \data = \op x + \eta,
\end{equation}
where $\op \in \BR^{d_{y} \times d_{x}}$ is a linear forward operator, and $\eta$ i.i.d.\ noise. We assume Gaussian noise $\eta \sim \mathcal{N}( 0, \sigma_{y}^{2} \mI)$. Many tomographic reconstruction problems take this form, e.g. computed tomography (CT).
Due to the inherent ill-posedness of the reconstruction task, e.g.\ $d_{y} \ll d_{x}$, 
suitable regularisation, or prior specification, is crucial for the successful recovery of $x$ \cite{tikhonov1977solutions, engl1996regularization, ito2014inverse}. 

In recent years, deep-learning based approaches have achieved outstanding performance on a wide variety of tomographic problems \cite{arridge2019solving, ongie2020deep, WangYe:2020}. 
Most deep learning methods are supervised; they rely on large volumes of paired training data. Alas, these often fail to generalise out-of-distribution \cite{antun2020instabilities}; small deviations from the distribution of the training data can lead to severe reconstruction artefacts. Pathologies of this sort motivate the need for more reliable alternatives \cite{WangYe:2020}. This work takes steps in this direction by exploring the intersection of unsupervised deep learning---not dependent on training data and thus mitigating hallucinatory artefacts \cite{BoraJalal:2017,HeckelHand:2019,Tolle2021meanfield}---and uncertainty quantification for tomographic reconstruction \cite{kompa2021second,Vasconcelos2022UncertaINR}.

\begin{figure}[t]
    \centering
    \includegraphics[width=0.9\columnwidth]{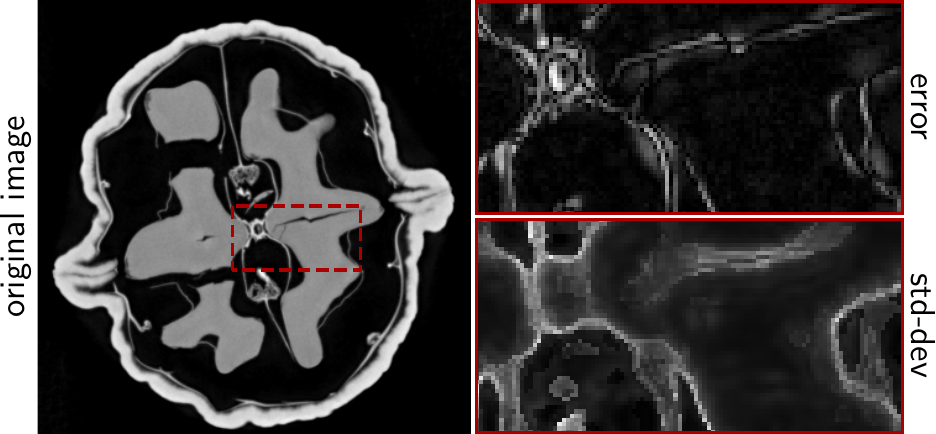}
    \caption{X-ray reconstruction ($501{\times}501\, \text{px}^2$) of a walnut (left), the absolute error of its CT reconstruction (top) and pixel-wise uncertainty (bottom).}
    \vspace{-0.5cm}
    \label{fig:walnut_mini}
\end{figure}

We focus on the deep image prior (DIP), perhaps the most widely adopted \cite{Ulyanov:2018} unsupervised deep learning approach.
DIP regularises the reconstructed image  $\hat x$ by reparametrising it as the output of a deep convolutional neural network (CNN). It relies solely on structural biases induced by the CNN architecture, and does not require paired training data. 
This idea has proven effective on tasks ranging from denoising and deblurring to more challenging tomographic reconstructions \cite{Ulyanov:2018,LiuSunXuKamilov:2019,baguer2020diptv,KnoppGrosser:2022,DarestaniHeckel:2021,GongLi:2019,cui2021populational,BarutcuGursoyKatsaggelos:2022}.
Nonetheless, the DIP only provides point reconstructions, and does not give uncertainty~estimates.

In this work, we equip DIP reconstructions with reliable uncertainty estimates. Distinctly from previous probabilistic formulations of the DIP \cite{Cheng2019Bayesian,Tolle2021meanfield}, we only estimate the uncertainty associated with a specific reconstruction, instead of trying to characterise a full posterior distribution over all candidate images. We achieve this by computing Gaussian-linear model type error-bars for a local linearisation of the DIP around its optimised reconstruction \cite{Mackay1992Thesis,Khan19approximate,Immer2021improving}.  Henceforth, we will refer to our method as \textit{linearised DIP}. Linearised approaches have recently been shown to provide state-of-the-art uncertainty estimates for supervised deep learning models \cite{daxberger2021bayesian}.
Unfortunately, the total variation (TV) regulariser, ubiquitous in CT reconstruction, makes inference in the linearised DIP intractable and it does not lend itself to standard Laplace (i.e. local Gaussian) approximations \cite{helin2022edgepromoting}. We tackle this issue by using the predictive complexity prior (PredCP) framework  \cite{nalisnick2021predicitve} to construct covariance kernels that induce properties similar to those of the TV prior while preserving Gaussian-linear conjugacy.
Finally, we discuss a number of computational techniques that allow us to scale the proposed method to large DIP networks and high-resolution 2D reconstructions.

Empirically, our method's pixel-wise uncertainty estimates predict reconstruction errors more accurately than existing approaches to uncertainty estimation with the DIP, cf.~\cref{fig:main_kmnist}.
This is not at the expense of accuracy in reconstruction: the reconstruction obtained using the standard regularised DIP method \cite{baguer2020diptv} is preserved as the predictive mean, ensuring compatibility with advancements in DIP research.
We demonstrate our approach on high-resolution CT reconstructions of real-measured 2D $\mu$CT projection data, cf.~\cref{fig:walnut_mini}.

The contributions of this work can be summarised as follows.
\begin{itemize}
    \item We propose a novel approach to bestow reconstructions from the TV-regularised-DIP with uncertainty estimates. Specifically, we construct a local linear model by linearising the DIP around its optimised reconstruction and provide this model's error-bars as a surrogate for those of the DIP. 
    \item We detail an efficient implementation of our method, scaling up to real-measured high-resolution $\mu$CT data. In this setting, our method provides by far more accurate uncertainty estimation than previous probabilistic formulations of the DIP reconstruction.
\end{itemize}

The rest of this paper is organised as follows. 
Related work is presented in \cref{sec:related_work}. \Cref{sec:background} introduces preliminaries necessary for our subsequent discussion of the linearised DIP. 
\Cref{sec:TV_prior} discusses how to design a tractable Gaussian prior based on the TV semi-norm for Bayesian inversion of ill-posed imaging problems.
\Cref{sec:probabilistic_model} and \cref{sec:computations} present the linearised DIP along with its efficient implementation.
\Cref{sec:experiments} reports our experimental investigations carried out on synthetic and real-measured high-resolution $\mu$CT data. \Cref{sec:conclusion} concludes the article. The detailed derivations and additional experimental results are given in the supplementary material (SM).

\section{Related Work}\label{sec:related_work}

\subsection{Advances in the deep image prior}

Since its introduction by \cite{Ulyanov:2018,ulyanov2020dip}, the DIP has been improved with early stopping \cite{Wang2021early}, TV regularisation \cite{LiuSunXuKamilov:2019,baguer2020diptv} and pretraining  \cite{barbano2021deep,KnoppGrosser:2022}.
We build upon these recent advancements by providing a scalable method to estimate the error-bars of TV-regularised-DIP reconstructions. 

Obtaining reliable uncertainty estimates for DIP reconstructions is a relatively unexplored topic.
Building upon \cite{aga2018cnngp} and \cite{Novak19CNN}, \cite{Cheng2019Bayesian} show that in the infinite-channel limit, the DIP converges to a Gaussian Process (GP). In the finite-channel regime, the authors
approximate the posterior distribution over DIP weights with Stochastic Gradient Langevin Dynamics (SGLD) \cite{Welling11Langevin}.
\cite{laves2020uncertainty} and \cite{Tolle2021meanfield} use factorised Gaussian variational inference \cite{Blundell15weight} and MC dropout \cite{hron18dropout,Vasconcelos2022UncertaINR}, respectively. 
These probabilistic treatments of DIP primarily aim to prevent overfitting, as opposed to accurately estimating uncertainty.
While they can deliver uncertainty estimates, their quality tends to be poor.
In fact, obtaining reliable uncertainty estimates from deep-learning based approaches, like DIP, largely remains a challenging open problem \cite{Snoek19can, ashukha2020pitfalls,Foong20Approximate, antoran2020depth,BarbanoArridgeJinTanno:2021}.
In the present work, we perform uncertainty estimation by performing Bayesian inference with respect to the DIP model locally linearised around its optimised parameters. This is distinct from the aforementioned approaches in that we only model a local mode of the posterior distribution.

\subsection{Bayesian inference in linearised neural networks}

The Laplace method is first applied to deep learning in \cite{Mackay1992Thesis}. It has seen a recent popularisation as the best performing approach when it comes to Bayesian reasoning with neural networks \cite{daxberger2021bayesian,Daxberger2021redux}. Specifically,\cite{Khan19approximate} and \cite{Immer2021improving} show that the linearization step improves the quality of uncertainty estimates. \cite{Immer21Selection}, \cite{antoran2021fixing} and \cite{antoran2022Adapting} explore the linear model's evidence for model selection. \cite{daxberger2021bayesian} and \cite{Maddox2021fast} introduce subnetwork and finite differences approaches, respectively, for scalable inference with linearised models. Inference in the linearised model is highly attractive compared to alternative approaches because it is post-hoc and it preserves the reconstruction obtained through DIP optimisation as the predictive mean.  
This line of work is also related to the neural tangent kernel \cite{Jacot2018NTK, Lee2019wide, Novak2020neuraltangents}, in which NNs are linearised at initialisation. 

\section{Preliminaries}\label{sec:background}

\subsection{Total variation regularisation}

The imaging problem given in \cref{eq:inverse_problem} admits multiple solutions consistent with the observation $\data$. Thus, regularisation is needed for stable reconstruction.
Total variation (TV) is perhaps the most well established  regulariser \cite{rudin1992nonlinear,chambolle2010introduction}. The anisotropic TV semi-norm of an image vector $x \in \BR^{d_{x}}$ imposes an $L^1$ constraint on image gradients,
\begin{align}\label{eq:TV_equation}
    \text{TV}(x) \!=\! \sum_{i,j} | X_{i,j} - X_{i+1,j}| + \sum_{i,j} |X_{i,j} - X_{i,j+1}|,
\end{align}
where  $X \in \BR^{h \times w}$ denotes the vector $x$ reshaped into an image of height $h$ by width $w$, and $d_{x} = h\cdot w$.
This leads to the regularised reconstruction formulation
\begin{equation}\label{eq:simple_optimisation_objective}
\hat{x} \in\argmin_{x\in\mathbb{R}^{d_x}} \|\op x - \data\|^2_2 + \lambda\mathrm{TV}(x),
\end{equation}
where the hyperparameter $\lambda>0$ determines the strength of the regularisation relative to the reconstruction term.

\subsection{Bayesian inference for inverse problems}

The probabilistic framework provides a consistent approach to uncertainty estimation in ill-posed imaging problems \cite{KaipioSomersalo:2005,Stuart,Seeger11Sparse}. In this framework, the image to be reconstructed is treated as a random variable.
Instead of finding a single best reconstruction $\hat{x}$, we aim to find a posterior distribution $p(x | \data)$ that scores every candidate image $x\in\BR^{d_{x}}$ according to its agreement with our observations $\data$ and prior beliefs $p(x)$. Under this view, the regularised objective in \cref{eq:simple_optimisation_objective} can be understood as the negative log of an unnormalised posterior, i.e. $p(x | \data) \!\propto\! \exp(- \mathcal{L}(x))$, and $\hat x$ as its mode, i.e. the \textit{maximum a posteriori} estimate. Specifically, the least squares reconstruction loss corresponds to the negative log of a Gaussian likelihood  $p(\data | x) = \mathcal{N}(\data; \op x, \mI)$ and the TV regulariser to the negative log of a prior density over reconstructions $p(x) \propto \exp(-\lambda\TV(x))$.

The aforementioned posterior is obtained by updating the prior over images with the likelihood as
\begin{gather}\label{eq:bayes_rule}
    p(x | \data) = p(\data)^{-1} p(\data | x)p(x),
\end{gather}
for $p(\data)= \int p(\data | x) p(x) {\rm d} x$ the normalising constant, also know as the marginal likelihood (MLL). This later quantity provides an objective for optimising hyperparameters, such as the regularisation strength $\lambda$.
The discrepancy among plausible reconstructions in the posterior indicates uncertainty.

This work partially departs from the established framework in that it \emph{solely concerns itself with characterising plausible reconstructions around the mode $\hat{x}$} \cite{Mackay1992Thesis}. 
This has two key advantages, i) when using NN models, the likelihood becomes a strongly multi-modal function, which in turn makes the full Bayesian update intractable. On the other hand, the posterior for the linearised model is quadratic; ii) even if we could obtain the full Bayesian posterior, downstream stakeholders without expertise in probability are likely to have little use for it. Instead, a single reconstruction together with its pixel-wise uncertainty may be more interpretable to end-users.

\subsection{The Deep Image Prior (DIP)} \label{subsec:preliminaries_DIP}

The DIP \cite{Ulyanov:2018,ulyanov2020dip} reparametrises the reconstructed image as the output of a CNN $x(\params)$ with a fixed input, which we omit from our notation for clarity, and learnable parameters $\params \in\mathbb{R}^{d_{\theta}}$ (we overload $x$ to refer to a point in $\BR^{d_{x}}$ and the injection $x: \,\BR^{d_{\theta}} \to \BR^{d_{x}}$ given by the DIP neural network).
This has been found to provide a favourable loss landscape for optimisation \cite{Shi2022Spectral}.
Penalising the TV of the DIP's output avoids the need for early stopping and improves reconstruction fidelity \cite{LiuSunXuKamilov:2019,baguer2020diptv}. 
The resulting optimisation problem is
\begin{equation}\label{eq:DIP_MAP-obj2}
    \hat \params \in\argmin_{\params\in\mathbb{R}^{d_\theta}} \|\op \net - \data\|_2^2 + \lambda \TV(\net),
\end{equation}
and the recovered image is given by $\hat x =\netmap$.
U-Net is the standard choice of CNN architecture \cite{ronneberger2015u}.  The DIP parameters must be optimised separately for each new measurement $\data$. Fortunately, the cost of this optimisation can be greatly lessened by pretraining \cite{barbano2021deep,KnoppGrosser:2022}.

\subsection{Bayesian inference with linearised neural networks}

Adopting an NN parametrisation of the reconstructed image, as in \cref{subsec:preliminaries_DIP}, makes the Bayesian posterior distribution in \cref{eq:bayes_rule} intractable. %
Instead, in the present work, we are only interested in the uncertainty associated with a specific regularised reconstruction $\hat x$, obtained per \cref{eq:DIP_MAP-obj2}.
To this end, we take a first-order Taylor expansion of the CNN function $x(\theta)$ around its optimised parameters $\hat \theta$ \cite{Mackay1992Thesis,Khan19approximate,Immer2021improving},
\begin{gather}\label{eq:linearised_model}
    h(\params) \coloneqq \netmap + \mJ(\params - \hat \params),
\end{gather}
where $\mJ \! \coloneqq \! \frac{\partial \net}{\partial \params}|_{\params = \hat\params} \! \in \! \BR^{d_{x} \times d_{\theta}}$ is the Jacobian of the CNN function with respect to its parameters evaluated at $\hat \params$.
\emph{We use this tangent linear model to obtain error-bars} for the DIP reconstruction $x(\hat\theta)$, while keeping $ x(\hat\theta)$ as the predictive mean.

When the observation noise is Gaussian, as is the case for the inverse problem formulation in \cref{eq:inverse_problem}, and a Gaussian prior is placed on $\theta$, we recover the conjugate setting; the posterior distribution over the linearised model's reconstructions is Gaussian with predictive covariance $\Sigma_{x|y} = J \Sigma_{\theta|y} J^\top$, where the posterior covariance over the parameters is given by  $\Sigma_{\theta|y} = (\mJ^\top \op^\top \op \mJ + \Sigma_{\theta\theta}^{-1})^{-1}$, for $\Sigma_{\theta\theta} \in \BR^{d_\theta \times d_\theta}$ the covariance of the Gaussian prior.
Combining the Gaussian linear model's error-bars with the DIP reconstruction, our predictive distribution is the Gaussian $\mathcal{N}(\hat x, \Sigma_{x|y})$.
In this conjugate setting, the marginal likelihood of the linearised model can also be computed in closed form and used to tune hyperparameters \cite{Mackay1992Thesis,Immer21Selection,antoran2021fixing,antoran2022Adapting}.

Computing the posterior covariance over linear model parameters $\Sigma_{\theta|y}$ has a cost scaling as $\mcO(d_{\theta}^{3})$.
For large CNNs, such as the U-Nets, this is intractable \cite{daxberger2021bayesian}.
In \cref{sec:probabilistic_model} and \cref{sec:computations}, we derive a dual (observation space) approach with a cost scaling as $\mcO(d_{y}^{3})$ and detail approaches to efficient implementation.
Furthermore, when conjugacy is lost, for instance due to the prior being non-Gaussian, the Laplace approximation is often used to approximate the posterior with a Gaussian surrogate. However, as we will see in \cref{sec:TV_prior}, the TV regulariser does not admit a Laplace approximation. \Cref{sec:TV_prior} and \cref{sec:probabilistic_model} are dedicated to addressing this issue.  

\section{The total variation as a prior}\label{sec:TV_prior}

The regularised reconstruction objective presented in \cref{eq:simple_optimisation_objective} can be interpreted as the negative log of an unormalised posterior distribution over images. 
In this context, the TV regulariser in \cref{eq:TV_equation} corresponds to a prior over candidate reconstructions of the form
\begin{gather}\label{eq:generic_tv_prior}
   p(x) = Z_{\lambda}^{-1} \exp({-\lambda\TV(x)}),
\end{gather}
where $Z_{\lambda} = \int \exp(-\lambda\TV(x))\, {\rm d} x$ is its normalisation constant (the prior is improper, since constant vectors are in the null space of the derivative operator). Working with this prior is computationally intractable because $Z_{\lambda}$ does not admit a closed form.
The Laplace method, which consists of a locally quadratic approximation, does not solve the issue because the second derivative of the TV regulariser is zero everywhere it is defined.

As an alternative to enforce local smoothness in the reconstruction, we construct a Gaussian prior $\mathcal{N}(0, \Sigma_{xx})$ with covariance $\Sigma_{xx}\in \BR^{d_{x} \times d_{x}}$ given by the Matern-$\nicefrac{1}{2}$ kernel
\begin{gather}\label{eq:Gaussian_over_x}
     [\Sigma_{xx}]_{ij,i'j'} = \sigma^{2}\exp\left({\frac{-\dist(i-i', j-j')}{\ell}}\right),
\end{gather}
where $i, j$ index the spatial locations of pixels of $x$, as in \cref{eq:TV_equation}, and $\dist(a, b) = \sqrt{a^{2} + b^{2}}$ denotes the Euclidean vector norm. The hyperparameter $\sigma^{2} \in \BR^{+}$ informs the pixel amplitude while the lengthscale parameter $\ell \in \BR^{+}$ determines the correlation strength between nearby pixels. 

\begin{figure*}[t]
    \centering
    \includegraphics[width=\textwidth]{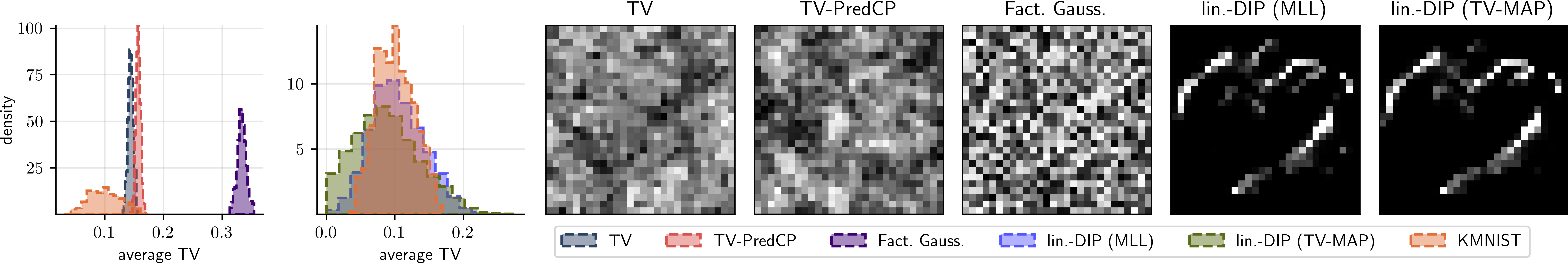}
    \vspace{-.6cm}
    \caption{Samples from priors. From left to right. Plot 1 shows a histogram of the average sample TV reporting an overlap between the TV and TV-PredCP priors. The factorised Gaussian prior results in larger TV values. Plot 2 shows an analogous histogram using samples from the linearised DIP (lin.-DIP) fitted to a KMNIST image, where the hyperparameters ($\ell$, $\sigma^{2}$) have been optimised both with and without the TV-PredCP term. Plots 3-5 show samples from the TV, TV-PredCP, and factorised Gaussian priors proposed in \cref{sec:TV_prior}, drawn using Hamiltonian Monte Carlo (HMC). Qualitatively, the TV prior produces samples with more correlated nearby pixel values than the factorised prior.  The TV-PredCP prior captures this effect and produces even smoother samples, likely due to the presence of longer range correlation in the Matern-$\nicefrac{1}{2}$ covariance. Plots 6-7 show prior samples from the linearised-DIP, which produces samples containing the structure of the KMNIST image used to train the network. The TV-PredCP term in the DIP hyperparameter optimisation to lead to smoother samples with less artefacts.
    }
    \vspace{-0.35cm}
    \label{fig:samples_from_priors}
\end{figure*}

The expected TV associated with our Gaussian prior is
\begin{gather}\label{eq:exact_ETV}
    \kappa := \BE_{x\sim\mathcal{N}(x; 0, \Sigma_{xx})}[\TV(x)] = c  \sigma \sqrt{1 - \exp(-\ell^{-1})},
\end{gather}
with $c$ a constant. See the SM for a proof. Here and below, we may omit the dependence of $\kappa$ on $(\ell, \sigma^{2})$ from our notation for clarity.
For fixed pixel amplitude $\sigma^{2}$, the expected reconstruction TV $\kappa$ is a bijection of the lengthscale $\ell$.  
We leverage this fact within the PredCP framework of \cite{nalisnick2021predicitve} to construct a prior over $\ell$ that will favour reconstructions with low TV
\begin{gather}\label{eq:prior_ell}
     p(\ell) = \text{Exp}(\kappa) \left| \nicefrac{\partial \kappa}{\partial \ell} \right|.
\end{gather}
The hierarchical prior over images, given by
\begin{gather}
x|\ell \sim \mathcal{N}(x; 0, \Sigma_{xx}), \quad \ell \sim \text{Exp}(\kappa) \left| \nicefrac{\partial \kappa}{\partial \ell} \right| ,
\end{gather}
is Gaussian for fixed $\ell$, and thus the prior is conditionally conjugate to Gaussian-linear likelihoods. 

Empirically, \cref{fig:samples_from_priors} shows strong agreement between samples, drawn with Hamiltonian Monte Carlo, from the described TV-PredCP prior and the intractable TV prior, both qualitatively and in terms of distribution over image TV.

\section{The linearised DIP}\label{sec:probabilistic_model}

In this section, we build a probabilistic model that aims to characterise diverse plausible reconstructions around $\hat{\theta}$, a mode of the regularised DIP objective, which we assume to have obtained using \cref{eq:DIP_MAP-obj2}. \Cref{subsec:linearised_prior} describes the construction of a linearised surrogate for the DIP reconstruction. \Cref{subsec:posterior_predictive} describes how to compute the surrogate model's error-bars and use them to augment the DIP reconstruction. \Cref{subseq:tv_for_NN} discusses how we include the effects of TV regularisation into the surrogate model.  Finally, in \cref{subsec:marginal_likelihood}, we describe a strategy to choose the surrogate model's prior hyperparameters using a marginal likelihood~objective.

\subsection{From a prior over parameters to a prior over images}\label{subsec:linearised_prior}

Upon training the DIP network to an optimal TV-regularised setting $\hat{x} = x(\hat{\params})$ using \cref{eq:DIP_MAP-obj2}, we linearise the network around $\hat \theta$ by applying \cref{eq:linearised_model}, and obtain the {affine-in-$\theta$} function $h(\theta)$. The error-bars obtained from Bayesian inference with $h(\theta)$ will act as a surrogate for the uncertainty in $\hat{x}$. To this end, we build the following hierarchical model,
\begin{gather}  
    \data|\params \sim \mathcal{N}(\data; \op h(\params), \sigma_{y}^{2} \mI), \notag\\
    \params|\ell \sim \mathcal{N}(\params; 0, \Sigma_{\params\params}(\ell)), \quad \ell \sim p(\ell), \label{eq:Model_weight_space} 
\end{gather}
where we have placed a Gaussian prior over the parameters $\theta$ that, in turn, depends on the lengthscale $\ell$. The choice of distribution over the lengthscale will allow us to incorporate TV constraints into the computed error-bars, cf. \cref{subseq:tv_for_NN}. We have introduced the noise variance $\sigma_{y}^{2}$ as an additional hyperparameter which we will learn using the marginal likelihood (details in \cref{subsec:marginal_likelihood}).
Importantly, conditional on a value of $\ell$, this is a conjugate Gaussian-linear model and thus the posterior distributions over $\params$ and over reconstructions are Gaussian and have a closed form. 

To provide intuition about the linearised model, we push samples from $\params\sim \mathcal{N}(\theta; {0}, \Sigma_{\params\params})$, through $h$. 
The resulting samples are drawn from a Gaussian prior distribution over reconstructions with covariance $\Sigma_{xx} \in \BR^{d_{x} \times d_{x}}$ given by $\mJ\Sigma_{\params\params}\mJ^{\top}$.
Here, the Jacobian $\mJ$ introduces structure from the NN function around the linearisation point $\hat\theta$. Indeed, \cref{fig:samples_from_priors} shows that samples contain a large amount of the structure of the KMNIST character that the DIP was trained on. 

\subsection{Efficient posterior predictive computation}\label{subsec:posterior_predictive} 
We augment the DIP reconstruction $\hat{x}$ with Gaussian predictive error-bars computed with the linearised model $h$ described in \cref{eq:Model_weight_space}, yielding $\mathcal{N}(\hat x, \Sigma_{x|\data})$.
The posterior covariance $\Sigma_{x|\data}$ is given by the Sherman–Morrison–Woodbury (SMW) formula
\begin{align}\label{eq:posterior_predictive}
    \Sigma_{x|\data} &=  \mJ(\sigma_{y}^{-2} \mJ^\top \op^\top  \op \mJ + \Sigma_{\theta\theta}^{-1})^{-1} \mJ^\top \\ \notag
    &= \Sigma_{xx} - \Sigma_{x\data} \Sigma_{\data \data}^{-1} \Sigma_{x\data}^{\top},
\end{align}
where $\Sigma_{xx} = \mJ\Sigma_{\params\params}\mJ^{\top}$, $\Sigma_{xy} = \Sigma_{xx}A^{\top}$ and $\Sigma_{yy} = A\Sigma_{xx}A^{\top} + \sigma_{y}^2 I$. The constant-in-$\params$ terms in $h$ do not affect the uncertainty estimates, and thus the error-bars match those of the simple linear model $J\theta$.
Importantly, \cref{eq:posterior_predictive} depends on the inverse of the observation space covariance $\Sigma_{\data \data}^{-1}$, as opposed to the covariance over reconstructions, or parameters. 
Thus, \cref{eq:posterior_predictive} scales as $\mathcal{O}(d_{x} d_{y}^{2})$ as opposed to $\mathcal{O}(d_{x}^{3})$ or $\mathcal{O}(d_{\theta}^{3})$ for the more-standard-in-the-literature output (reconstruction) space or parameter space approaches, respectively \cite{Immer2021improving,Daxberger2021redux}. 

\subsection{Incorporating the TV-smoothness} \label{subseq:tv_for_NN}

We aim to impose constraints on $h$'s error-bars, such that our model only considers low TV reconstructions as plausible.
For this, we place a block-diagonal Matern-$\nicefrac{1}{2}$ covariance Gaussian prior on our linearised model's weights, similarly to \cite{fortuin2021bayesian}. 
Specifically, we introduce dependencies between parameters belonging to the same CNN convolutional filter~as
\begin{equation}\label{eq:CovFunc}
     [\Sigma_{\params\params}]_{kij, k'i'j^{'}} = 
            \sigma^{2}\exp\Big({\frac{-\dist(i-i', j-j^{'})}{\ell_{d}}}\Big)\delta_{kk'},
\end{equation} %
where $k$ indexes the convolutional filters in the CNN, $\delta_{kk'}$ denotes Kronecker symbol, and $(i,j)$ index the spatial locations of specific parameters within a filter. The lengthscale $\ell_{d}$ regulates the filter smoothness. 
Intuitively, an image generated from convolutions with smoother filters will present lower TV, and indeed \cref{fig:monotonicity_mini} shows a bijective relationship between this quantity and the filter lengthscale. The hyperparameter $\sigma^{2}_{d}$ determines the marginal prior variance. 

\begin{figure}
    \centering
    \includegraphics[width=0.85\columnwidth, trim={0.25cm 0 0 0},clip]{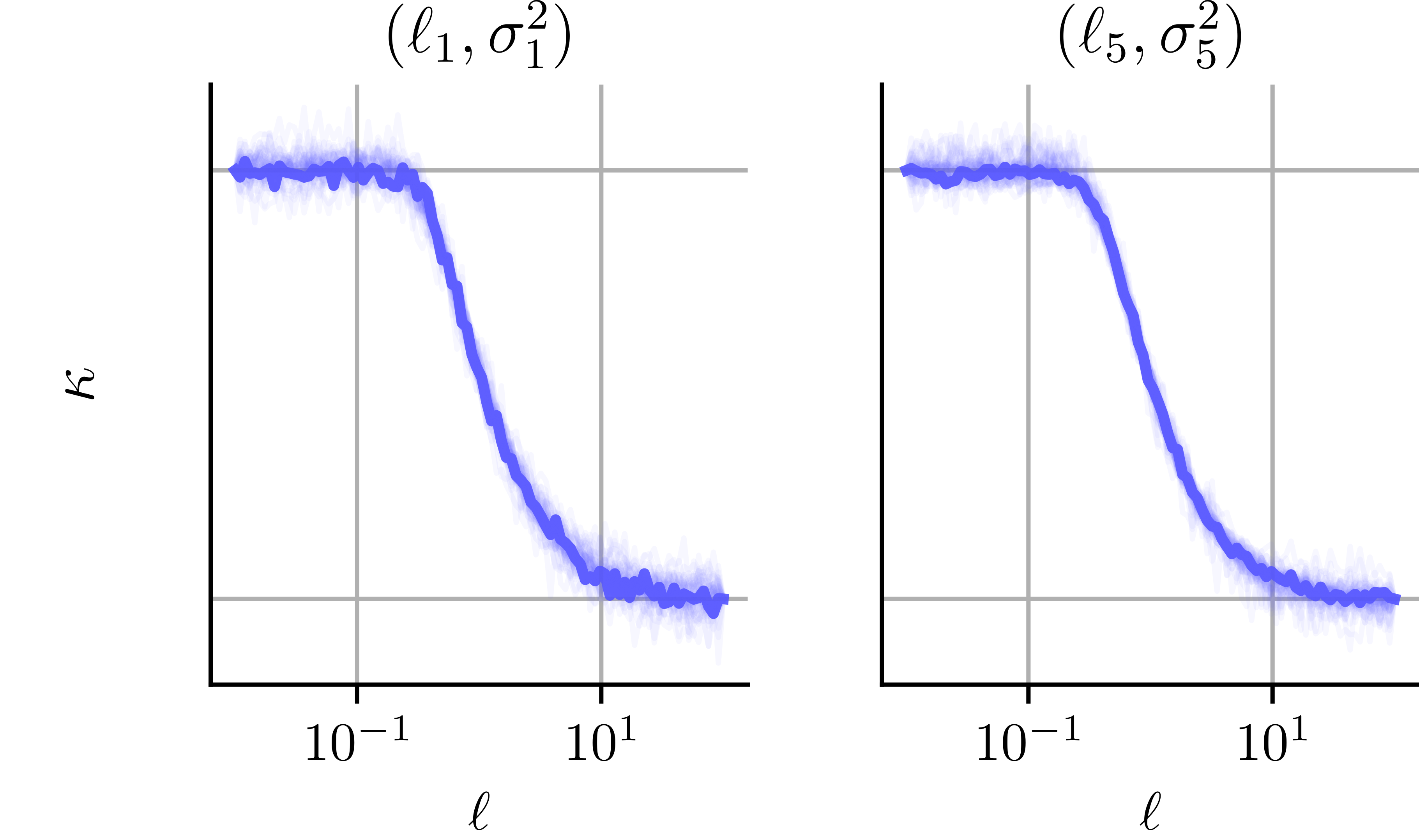}
    \vspace{-1em}
    \caption{Experimental evidence of monotonicity between $\ell$ and $\kappa$, implying the desired invertibility of the mapping, computed over 50 KMNIST test images for two blocks of the linearised 3-scale U-Net: the input block ($\ell_1, \sigma^2_1$) and the outermost $3\times 3$ convolutional block $(\ell_5, \sigma^2_5)$. 
    We draw 500 samples to estimate $\kappa$ for each value of $\ell$.}
    \vspace{-1.5em}
    \label{fig:monotonicity_mini}
\end{figure}

Both parameters are defined per architectural block $d \in \{1,2,  \ldots , D\}$ in the U-Net and we write $\ell = [\ell_{1}, \ell_{2}, \ldots , \ell_{D}]$ and $\sigma^{2} = [\sigma^{2}_{1}, \sigma^{2}_{2}, \ldots , \sigma^{2}_{D}]$. The chosen U-Net architecture is fully convolutional and thus \cref{eq:CovFunc} applies to all parameters, reducing to a diagonal covariance for $1\times 1$ convolutions.
A diagram of the architecture highlighting the described prior blocks $1, \ldots , D$ is in \cref{fig:unet-diagram}. 

To enforce TV-smoothness, we adopt the strategy outlined in \cref{sec:TV_prior}.
Since choosing a large $\ell$ enforces smoothness in the output, a prior placed over the filter lengthscales $\ell$ can act as a surrogate for the $\TV$ prior. To make this connection explicit, we construct a TV-PredCP \cite{nalisnick2021predicitve}
\begin{gather}\label{eq:ell_prior_DIP}
    p(\ell) = \prod_{d=1}^{D} p(\ell_{d}) =\prod_{d=1}^{D}\text{Exp}(\kappa_{d})\left|\frac{\partial \kappa_{d}}{\partial \ell_{d}}\right|,\\
    \text{with }\kappa_{d} := \BE_{
    \mathcal{N}(\hat \theta_d, \Sigma_{\theta_d\theta_d})
    \prod^{D}_{i=1,i\neq d}\delta(\params_{i} - \hat \params_{i})}\left[{\lambda \TV}({h}(\params)) \right],\label{eq:kappa_definition}
\end{gather}
where $\kappa_{d}$ is the expected $\TV$ of the CNN output over the prior uncertainty in the parameters of block $d$ when all other entries of $\params$ are fixed to $\hat \theta$. This latter choice ensures the mean reconstruction matches $\hat x$.
We relate the expected $\TV$ to the filter lengthscale $\ell_{d}$ by means of the change of variables formula. The separable-across-blocks form of $p(\ell)$ and consequent block-wise definition of $\kappa_{d}$ ensures dimensionality preservation, formally needed in the change of variables. By the triangle inequality, it can be verified that $\sum_{d} \kappa_{d}$ is an upper bound on the expectation under the joint distribution $\BE_{\mathcal{N}(\theta; \hat \theta, \Sigma_{\theta\theta})}[\TV({h}(\params))]$; see the SM for a proof.
Note that \cref{eq:ell_prior_DIP} can be computed analytically. However, its direct computation is costly and we instead rely on numerical methods described in \cref{sec:computations}. In \cref{fig:samples_from_priors} (cf. plot 2 and plots 6-7), we show samples from $\mathcal{N}(x; 0, \Sigma_{xx})$ where $\ell$ is chosen using the marginal likelihood with TV-PredCP constraints (discussed subsequently in \cref{subsec:marginal_likelihood}).
Incorporating the TV-PredCP leads to smoother samples with less discontinuities.

\begin{figure}[h]
    \centering
    \includegraphics[width=\columnwidth]{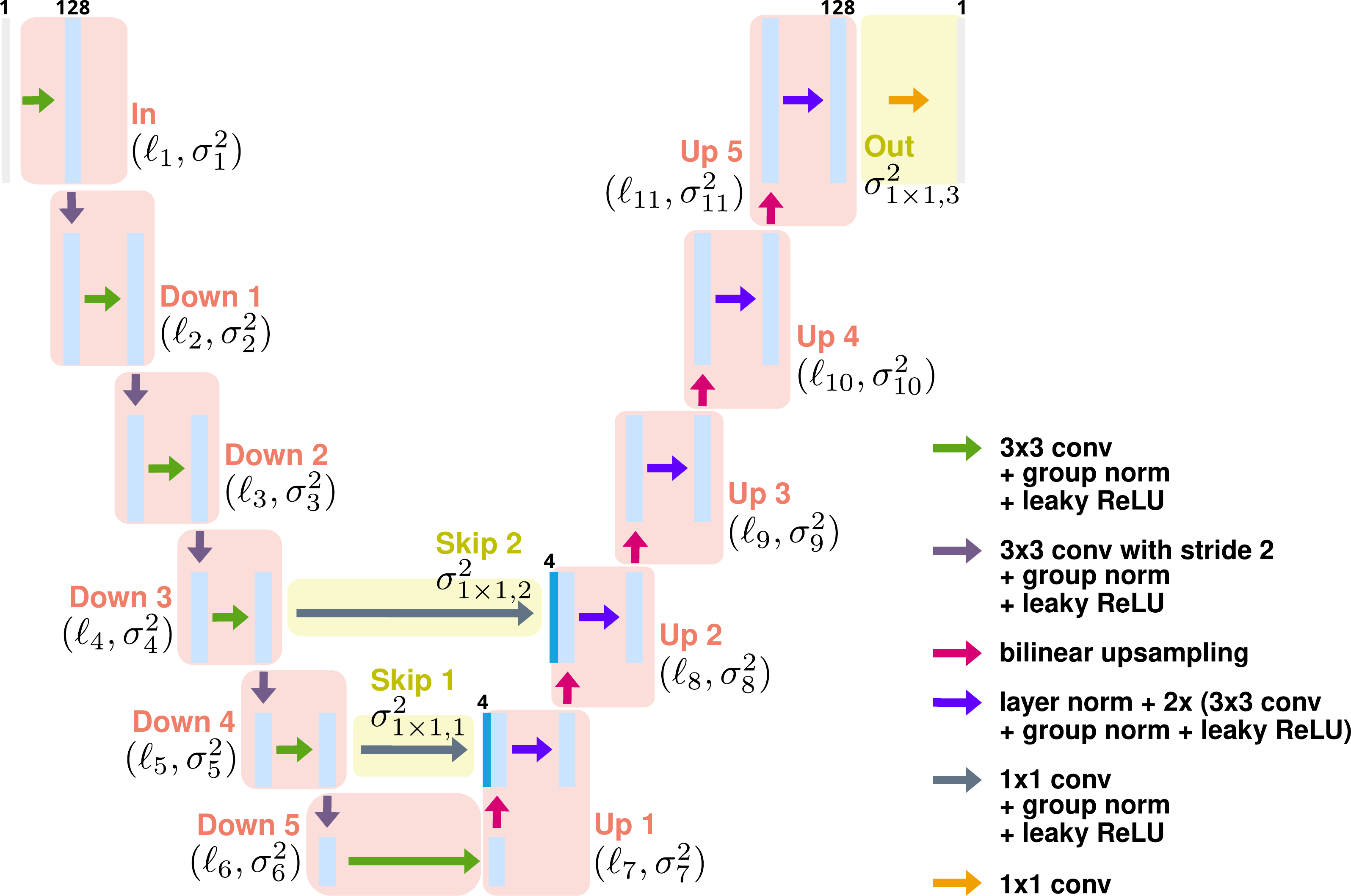}
    \vspace{-1.5em}
    \caption{A schematic illustration of the U-Net architecture used in the numerical experiments on Walnut data (see \cref{subsec:walnut-exps}). For KMNIST, we use a reduced, 3-scale U-Net without group norm layers (see \cref{fig:unet-diagram-kmnist}). Each light-blue rectangle corresponds to a multi-channel feature map. We highlight the architectural components corresponding to each block ${1,\dotsc,D}$ for which a separate prior is defined with red and yellow boxes.}
    \vspace{-0.75em}
    \label{fig:unet-diagram}
\end{figure}

\subsection{Type-II MAP learning of hyperparameters}\label{subsec:marginal_likelihood}

The calibration of our predictive Gaussian error-bars, as described above, depends crucially on the choice of the hyperparameters of the linearised hierarchical model given in \cref{eq:Model_weight_space} \cite{antoran2022Adapting}, that is, on  the values of $(\sigma^{2}_{y}, \sigma^{2}, \ell)$. For a given lengthscale $\ell$, Gaussian-linear conjugacy leads to a closed form marginal likelihood objective to learn these hyparameters. In turn, to learn $\ell$, we combine the aforementioned objective with the TV-PredCP's log-density, which acts as a regulariser. The resulting expression resembles a Type-II \emph{maximum a posteriori} (MAP) \cite{gpbook} objective
{
\vspace{-.1cm}
\begin{align}
    &\log  \,p(\data |\ell; \sigma^{2}_{y}, \sigma^{2} ) + 
    \log p(\ell; \sigma^{2}) \approx \notag
    \\
    &- \frac{1}{2} \sigma^{-2}_{y}||\data - \op\netmap||_{2}^{2} - \frac{1}{2} \hat \params_{h}^{\top} \Sigma^{-1}_{\params\params}(\ell, \sigma^{2}) \hat \params_{h} \label{eq:type2MAP}\\
    &- \frac{1}{2} \log |\Sigma_{\data \data}| - \sum_{d=1}^{D} \kappa_{d}(\ell, \sigma^2) + \log \left|\frac{\partial \kappa_{d}(\ell, \sigma^2)}{\partial \ell_{d}}\right| +  B, \notag
\end{align}}where $B$ is a constant independent of the hyperparameters and the vector $\hat \params_{h} \in \BR^{d_{\params}}$ is the posterior mean of the linear model's parameters. See the SM for the detailed derivation of the expression. Following \cite{antoran2021fixing}, we compute this vector by minimising $\sigma_{y}^{-2}\|\op h(\theta_{h}) - y\|_2^2 + \lambda \TV({h}(\params_h))$. It is worth noting that the vector $\hat\theta_h$ did not appear in our predictive distribution over reconstructions \cref{eq:posterior_predictive}, since we use the DIP reconstruction as the predictive mean. 
The bottleneck in evaluating \cref{eq:type2MAP} is the log-determinant of the observation covariance $\Sigma_{yy}$, which has a cost $\mathcal{O}(d_{y}^{3})$.
In the following section we describe scalable ways to approximate the log-determinant together with other expensive to compute quantities required for prediction.

\section{Towards scalable computation}\label{sec:computations}

In a typical tomography setting, the dimensionality $d_x$ of the reconstructed image $\hat x$ and $d_y$ of the observation $\data$ can be large, e.g. $d_{x} > {\rm1e5}$ and $d_{y} >{\rm 1e3}$. Due to the former, holding in memory the input space covariance matrices (e.g.\ $\Sigma_{xx}$ and $\Sigma_{x|\data}$) is infeasible.   The latter greatly complicates the computation of the log-determinant of the observation space covariance $\Sigma_{\data \data}$ in the marginal likelihood \cref{eq:type2MAP} (or its gradients), and its inverse in the posterior predictive distribution \cref{eq:posterior_predictive}, which scale as $\mathcal{O}(d_{y}^{3})$ and $\mathcal{O}(d_{x} d_{y}^{2})$, respectively. To scale our approach to tomographic problems, throughout this paper, we only access Jacobian and covariance matrices through matrix–vector products, commonly known as \textit{matvecs}. Specifically, our workhorses are products resembling ${v}_{x}^{\top} \Sigma_{xx}$ and ${v}_{y}^{\top} \Sigma_{\data \data}$ for ${v}_{x} \in \BR^{d_{x}}$ and ${v}_{y} \in \BR^{d_{y}}$. We compute ${v}_{y}^{\top} \Sigma_{\data \data}$ through successive matvecs with the components of $\Sigma_{\data \data}$
\begin{gather}\label{eq:closure}
    {v}_{y}^{\top} \Sigma_{\data \data} = {v}_{y}^{\top} \left(\op \mJ \Sigma_{\params\params} \mJ^{\top} \op^{\top} + \sigma_{y}^{2} \mI_{d_{y}}\right),
\end{gather}
and we compute ${v}_{x}^{\top} \Sigma_{xx}$ similarly. 
We compute Jacobian vector products ${v}_{\theta}^{\top} \mJ^{\top}$ for $v_\theta\in \mathbb{R}^{d_\theta}$ using forward mode automatic differentiation (AD) and ${v}_{x}^{\top} \mJ$ using backward mode AD, using the \texttt{functorch} library \cite{functorch2021}.
We compute products with $\Sigma_{\params\params}$ by exploiting its block diagonal structure.
All these operations can be batched using modern numerical libraries (i.e. \texttt{functorch}, \texttt{GpyTorch} \cite{Gardner2018gpytorch}, \texttt{PyTorch})~and~GPUs.

\subsection{Conjugate gradient log-determinant gradients}
\label{subsec:cg_logdet_grads}

For the TypeII-MAP optimisation in \cref{eq:type2MAP}, we estimate the gradients of $\log|\Sigma_{\data \data}|$ with respect to the parameters of interest $\phi$ using the stochastic trace estimator \cite{Hutchinson90trace}
\begin{align}\label{eq:cg_logdet_grad}
\frac{\partial \log|\Sigma_{\data \data}|}{\partial \phi} &= \text{Tr}\left( \Sigma_{\data \data}^{-1} \frac{\partial \Sigma_{\data \data}}{\partial \phi} \right) \\
&= \BE_{{v}\sim \mathcal{N}(v; 0, \mP)}\left[{v}^{\top} \Sigma_{\data \data}^{-1} \frac{\partial \Sigma_{\data \data}}{\partial \phi} \mP^{-1}{v}\right] \notag,
\end{align}
where $P$ is a preconditioner matrix.
We approximately solve the linear system ${v}^{\top} \Sigma_{\data \data}^{-1}$ for batches of probe vectors ${v}$ using the \texttt{GPyTorch} preconditioned conjugate gradient (PCG) implementation \cite{Dong2017determinant,Gardner2018gpytorch}.

Our preconditioner $\mP$ is constructed
using randomised SVD. 
Specifically, we approximate $\op \mJ\Sigma_{\params\params}\mJ^{\top}\op^{\top}$---for simplicity denoted as $H\in\mathbb{R}^{d_{y}\times d_{y}}$---as $ \tilde U \tilde \Lambda \tilde U^{\top}$, using a randomised eigendecomposition algorithm  \cite{Halko2011structure}, \cite{martinsson2020randomized}.
The approach first computes an orthonormal basis capturing the space spanned by $H$'s columns.
The idea is to obtain a matrix Q with $r$ orthonormal columns, that approximates the range of $H$.
This is done by constructing a standard normal test matrix $\Omega \in \mathbb{R}^{d_{y} \times r}$, and computing the (thin) QR decomposition of $H\Omega$.
Once $Q$ is computed, we solve for a symmetric matrix $B\in\mathbb{R}^{r\times r}$ (much smaller than $H$) such that $B$ approximately satisfies $B(Q^\top\Omega) \approx Q^\top H\Omega$.
We then compute the eigendecomposition of $B$, $V\Lambda V^\top$, and recover $\tilde U = QV$.
This method requires $\mathcal{O}(r)$ matvecs resembling $Hv$ to construct not only an approximate basis but also its complete factorisation. 
Finally, the preconditioner $\mP$ is defined as $\tilde U \tilde \Lambda \tilde U^{\top} + \sigma^{2}_{y}\mI$.
To compute $\mP^{-1}v$ efficiently, we make use of the SMW formula. 
Since $\mP$ depends on $\phi$, we interleave the updates of $\mP$ with the optimisation of \cref{eq:type2MAP}.

\subsection{Ancestral sampling for TV-PredCP optimisation}\label{sec:predcptv_computation}

For large images, exact evaluation of \cref{eq:kappa_definition} for the expected TV is computationally intractable. Instead, we estimate the gradient of $\kappa_d$ with respect to the parameters of interest $\phi = (\sigma^2, \ell)$ using a Monte-Carlo approximation to 
\begin{gather}
    \frac{\partial \kappa_{d}}{\partial\phi} = \BE_{\params_{d}\sim \mathcal{N}(\theta_d; \hat \theta_{d}, \Sigma_{\theta_d \theta_d})}\left[ \frac{\partial \TV(x)}{\partial x} \mJ_{d}\frac{\partial \params_{d} }{\partial \phi}\right],
    \label{eq:tv_grad_est}
\end{gather}
where $\frac{\partial \TV(x)}{\partial x}$ is evaluated at the sample $x{=}\mJ_{d}\params_{d}$ and $\frac{\partial \params_{d}}{\partial \phi}$ is the reparametrisation gradient for $\params_{d}$, a prior sample of the weights of CNN block $d$. Since the second derivative of the $\TV$ semi-norm is almost everywhere zero, the gradient for the change of variables volume ratio is
\begin{gather}\label{eq:tv_grad_est_2}
    \frac{\partial^2 \kappa_{d}}{\partial\phi^2} = \BE_{\params_{d}\sim \mathcal{N}(\theta_d; \hat \theta_{d}, \Sigma_{\theta_d \theta_d})}\left[ \frac{\partial \TV(x)}{\partial \phi} \mJ_{d}
    \frac{\partial^2 \params_{d}}{\partial \phi^2}
    \right].
\end{gather}

\subsection{Posterior covariance matrix estimation by sampling}\label{sec:post_sampling}

The covariance matrix $\Sigma_{x | \data}$ is too large to fit into memory for most tomographic reconstructions. Instead, we follow \cite{Wilson21pathwise} in drawing samples from $\mathcal{N}(x; 0, \Sigma_{x|\data})$ via Matheron's rule
\begin{gather}
    x_{x|\data} = x + \Sigma_{xy}\Sigma_{\data \data}^{-1}(\epsilon - \op  x_0); \notag \\
    x_0 = \mJ \params_0; \quad \params_0 \sim \mathcal{N}(\theta; 0, \Sigma_{\params\params}); \quad \epsilon \sim \mathcal{N}(\epsilon; 0, \sigma_{y}^{2} \mI).\label{eq:matheron}
\end{gather}
The biggest cost lies in constructing $\Sigma_{\data \data}$, which is achieved by applying \cref{eq:closure} to the standard basis vectors $\Sigma_{\data \data} = [e_{1}, e_{2},\,...\, e_{d_{y}}]^{\top} \Sigma_{\data \data}$. We then perform its Cholesky factorisation as an intermediate step towards matrix inversion, both relatively costly operations.
Fortunately, we only have to repeat these once, after which the sampling step in \cref{eq:matheron} can be evaluated cheaply.
Alternatively, as in \cref{eq:cg_logdet_grad}, we can compute the solution of the linear system, $\Sigma_{\data\data}^{-1} v_{\data}$ for any $v_{\data}$ via PCG, without explicitly assembling (thus storing in memory) the measurement covariance matrix, or computing its Cholesky factorisation.
This approach allows us to scale the sampling operation to large measurement spaces, where the matrix $\Sigma_{\data\data}$ may not fit in memory.

The samples drawn are zero mean, as the quantities we are interested in do not depend on the linear model's mean.
The full predictive covariance matrix $\Sigma_{x|y}$ does not fit in memory. However, we only expect our predictions to be correlated for nearby pixels. Thus, we estimate cross covariances for patches of only up to $10\times 10$ adjacent pixels using the stabilised formulation of \cite{Maddox2019Simple}: $\hat \Sigma_{x|y} = \frac{1}{2k} \left[\sum_{j=1}^k {x_{j}}^{2} + x_{j} {x_{j}}^{\top}\right]$ for $({x}_{j})_{j=1}^k$ samples from the posterior predictive over a patch. Using larger patches yields little to no improvements.

\begin{figure*}[thb]
    \centering
    \includegraphics[width=\textwidth]{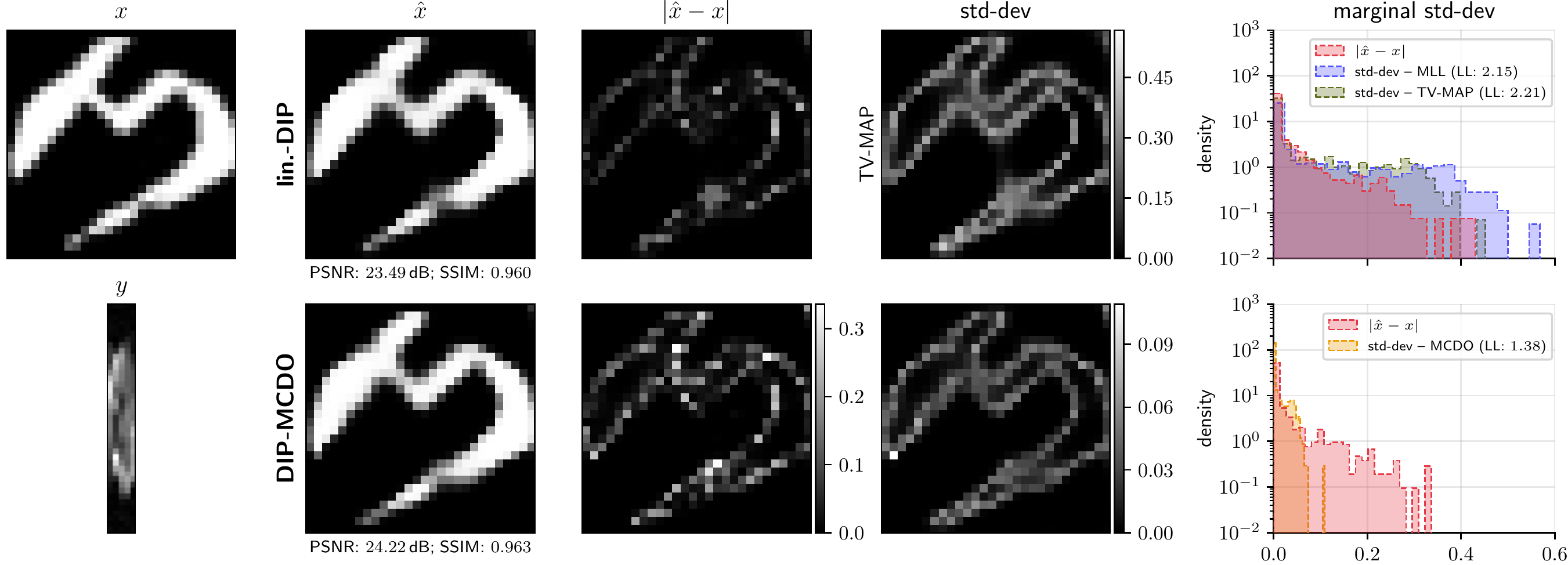}
    \vspace{-.5cm}
    \caption{Exemplary character recovered from $y$ (using $5$ angles and $\eta(5\%)$) with lin.-DIP and DIP-MCDO along with respective uncertainty estimates.}
    \vspace{-1.25em}
    \label{fig:main_kmnist}
\end{figure*}

We now turn to accelerating Jacobian matvecs through approximate computations.
\Cref{tab:wall-clock} shows that the Jacobian matvecs involved in sampling from the posterior predictive (2$\times$ $v_{\theta}^{\top}\mJ^{\top}$ + 1$\times$ $v_{x}^{\top}\mJ$) take $\approx100\,\%$ of this step's computation time (2.4\,h).
To address this inefficiency, we construct a low-rank approximation $\tilde \mJ$ %
of the Jacobian matrix $\mJ$, which we store in memory. $v_{\theta}^{\top}\tilde\mJ^{\top}$ and $v_{x}^{\top}\tilde\mJ$ can then be computed via matrix multiplication (as opposed to automatic differentiation), which is a highly optimised primitive. This allows for a fast yet approximate computation of %
matvecs with $\Sigma_{\data\data}$ by substituting $\tilde J$ into \cref{eq:closure}. In turn, this results in faster sampling from the posterior predictive, bring it from 2.4 hours down to less than a minute.
We construct $\tilde J$ similarly to $\mP$ in \cref{subsec:cg_logdet_grads}. That is, following  \cite{Halko2011structure}, we build a structured $r$-rank approximation to $\mJ$, by having access only to matvecs~with~$\mJ$~and~$\mJ^{\top}$.

\begin{table}[h]
    \centering
    \caption{Wall-clock time for the different steps of our algorith when applied to the Walnut data (see \cref{subsec:walnut-exps}) using A100 GPU. $d_{y}, d_{x}, d_{\theta}$ are n. observed pixels, reconstructed pixels, NN~params respectively. Computations reported below the dotted line are in double precision.}
    \resizebox{\linewidth}{!}{
    \begin{tabular}{lr@{\hskip 0.9in}}
        & wall-clock time\hspace{-0.4in} \\
    \hline
    DIP optim. (after pretraining \cite{barbano2021deep})& \hphantom{0}$\mathllap{<}$0.1\,h\\
    Hyperparam. optim. (MLL) & 26.2\,h\\
    Hyperparam. optim. (TV-MAP) & 35.4\,h\\[0.1cm]
    \hdashline%
    \addlinespace[0.1cm]%
    Assemble $\Sigma_{\data\data}$ & \hphantom{0}2.7\,h\\
    Draw 4096 posterior samples  & \hphantom{0}2.4\,h\\
    (Evaluate 4096 times 2$\times$ $v_{\theta}^{\top}\mJ^{\top}$ + 1$\times$ $v_{x}^{\top}\mJ$) & \hphantom{0}2.4\,h\\
    Draw 4096 posterior samples ($\tilde \mJ$ $\&$ PCG) & \hphantom{0}0.1\,h$\mathrlap{\text{ (+0.2\,h for $\tilde \mJ$)}}$\\
    (Evaluate 4096 times 2$\times$ $v_{\theta}^{\top}\tilde \mJ^{\top}$ + 1$\times$ $v_{x}^{\top}\tilde \mJ$) & \hphantom{0}$\mathllap{<}$0.1\,$\mathrlap{\text{min (+0.2\,h for $\tilde \mJ$)}}$\hphantom{h}\\
    \hline
    \end{tabular}
    }
    \vspace{-1em}
    \label{tab:wall-clock}
\end{table}

\section{Experimental evaluation} \label{sec:experiments}

In this section, we experimentally evaluate: i) the properties of the models and priors discussed in \cref{sec:TV_prior,sec:probabilistic_model}, and whether they lead to accurate reconstructions and calibrated uncertainty; ii) the fidelity of the approximations described in \cref{sec:computations}; and iii) the performance of the proposed method ``linearised-DIP'' (lin.-DIP) relative to the previous MC dropout (MCDO) based probabilistic formulation of DIP \cite{laves2020uncertainty}.
We attempted to include DIP-SGLD \cite{Cheng2019Bayesian} in our analysis, but were unable to get the method to produce competitive results on tomographic reconstruction problems. For each individual image to be reconstructed, we employ the following linearised DIP inference procedure: i) optimise the DIP weights via \cref{eq:DIP_MAP-obj2}, obtaining $\hat x = x(\hat\params)$; ii) optimise prior hyperparameters ($\sigma_{y}^{2}$, $\ell$, $\sigma^{2}$) via \cref{eq:type2MAP}; iii) assemble and Cholesky decompose $\Sigma_{\data \data}$ with \cref{eq:closure} (this step can be accelerated using approximate methods \cref{sec:post_sampling}); iv) compute posterior covariance matrices either via \cref{eq:posterior_predictive}, or estimate them via \cref{eq:matheron}. 

\subsection{Reconstruction of KMNIST digits}

Our initial analysis uses simulated CT data obtained by applying \cref{eq:inverse_problem} to 50 images from the test set of the Kuzushiji-MNIST (KMNIST) dataset: $28\times 28$ ($d_{x}=784$) grayscale images of Hiragana characters \cite{clanuwat2018deep}. For each image, we choose the noise standard deviation to be either 5\% or 10\% of the mean of $\op x$, denoted as $\eta(5 \%)$ or $\eta(10\%)$.
The forward operator $\op$ is taken to be the discrete Radon transform, assembled via \texttt{ODL} \cite{adler2017operator}, a commonly employed software package in CT reconstruction. For KMNIST we use a U-Net with depth of 3 and $76905$ parameters (a down-sized net compared to the one in \cref{fig:unet-diagram}).

\subsubsection{Comparing linearised DIP with network-free priors} \label{subsec:HMC-exps}

We first evaluate the priors described in \cref{sec:TV_prior}, that is, the intractable TV prior, the proposed TV-PredCP with a Matern-$\nicefrac{1}{2}$ kernel and a factorised Gaussian prior, by performing inference in the setting where the operator $\op$ collects $5$ angles ($d_{y}=205$) sampled uniformly from $0^\circ$ to $180^\circ$ and is applied to $50$ KMNIST test set images. Here, $10\%$ noise is added.
This results in a very ill-posed reconstruction problem, maximising the relevance of the prior.
We select the $\sigma_y^{2}$ and $\lambda$ hyperparameters for the factorised Gaussian prior and the intractable TV prior respectively such that the posterior mean's PSNR is maximised across a validation set of 10 images from the KMNIST training set. 
We keep the choice of $\sigma_y^{2}$ and $\lambda$ hyperparameters from the first two models for our experiments with the third model (Matern-$\nicefrac{1}{2}$ with TV-PredCP prior over $\ell$).
For all priors, we perform inference with the NUTS HMC sampler. We run 5 independent chains for each image. We burn these in for $3\times 10^3$ steps each and then proceed to draw $10^{4}$ samples with a thinning factor of 2.
We evaluate test log-likelihood using Gaussian Kernel Density Estimation (KDE) \cite{Silverman:1986}. The kernel bandwidth is chosen using cross-validation on 10 images from the training set.

The results in \cref{tab:HMC} show that the TV-PredCP performs best in terms of the test log-likelihood and both posterior mean and posterior mode PSNR, followed by the TV and then the factorised Gaussian. This is somewhat surprising considering that this prior was designed as an approximation to the intractable TV prior.
We hypothesise that this may be due to the Matern model allowing for faster transitions in the image than the TV prior, while still capturing local correlations, as shown qualitatively in \cref{fig:samples_from_priors}. This property may be well-suited to the KMNIST datasets, where most pixels either present large amplitudes or are close to 0. For comparison, we include results for DIP-based predictions, which handily outperform all non-NN-based methods.

\begin{table}[h]
\centering
\caption{Quantitative results for inference with the different priors introduced in \cref{sec:TV_prior}. Results for lin.-DIP prior are also provided to facilitate their comparison. We report both the PSNR of $\BE[x | \data]$, which denotes the posterior mean reconstruction, and the PSNR of $\hat x$, which denotes the posterior mode found through optimisation.}
\label{tab:HMC}
\resizebox{\columnwidth}{!}{%
\begin{tabular}{lccc}
 & log-likelihood & $\BE[x | \data]$ &  $\hat x$\\ \hline
Fact. Gauss. & $0.30 \pm 0.17$ & $16.15 \pm 0.38$ & $14.89 \pm 0.38$ \\
TV & $0.49 \pm 0.14$ & $16.32 \pm 0.38$ & $16.29 \pm 0.41$ \\
TV-PredCP & \textbf{0.65} $\pm$ \textbf{0.12} & \textbf{16.55} $\pm$ \textbf{0.39} & \textbf{17.48} $\pm$ \textbf{0.39} \\ \hline 
lin.-DIP (MLL) & $1.63 \pm 0.08$ & $-$ & 19.46 $\pm$ 0.52\\
lin.-DIP (TV-MAP) & $1.63 \pm 0.09$ & $-$ & 19.46 $\pm$ 0.52\\ 
\hline
\end{tabular}%
}
\vspace{-1.25em}
\end{table}

\subsubsection{Comparing calibration with DIP uncertainty quantification baselines}\label{sec:calibration_comparison}

Using KMNIST, we construct test cases of different ill-posedness by simulating the observation $\data$ with four different angle sub-sampling settings for the linear operator $\op$: $30$ ($d_{y}{=}1230$), $20$ ($d_{y}{=}820$), $10$ ($d_{y}{=}410$) and $5$ ($d_{y}{=}205$) angles are taken uniformly from the range $0^\circ$ to $180^\circ$.
We consider two noise configurations by adding either $5\%$ or $10\%$ noise to the exact data $\op x$.
We evaluate all DIP-based methods using the same $50$ randomly chosen KMNIST test set images.
To ensure a best-case showing of the methods, we choose appropriate hyperparameters for each number of angles and white noise percentage setting by applying grid-search cross-validation, using 50 images from the KMNIST training dataset.
Specifically, we tune the TV strength $\lambda$ and the number of iterations for linearised DIP. Due to the reduced image size, we apply linearised DIP as in \cref{sec:probabilistic_model}, without approximate computations. 
As an ablation study, we include additional baselines: linearised DIP without the TV-PredCP prior over hyperparameters (labelled MLL), and DIP reconstruction with a simple Gaussian noise model consisting of the back-projected observation noise $\mathcal{N}(x; \hat x, \sigma^{2}_{\op} \mI)$, with $\sigma^{2}_{\op} = \sigma^{2}_{y}\text{Tr}((\op^{\top}\op)^{\dagger}) d^{-1}_{x}$ where $\sigma_{y}^{2}{=}1$ (labelled $\sigma_{y}^{2}{=}1$).
Note that non-dropout methods share the same DIP parameters $\hat\params$, and thus the same mean reconstruction.
Hence, higher values in log-density indicate better uncertainty calibration, i.e., the predictive standard deviation better matches the empirical reconstruction error.
DIP-MCDO does not provide an explicit likelihood function over the reconstructed image. We model its uncertainty with a Gaussian predictive distribution with covariance estimated from $2^{14}$ samples. MNIST images are quantised to 256 bins, but our models make predictions over continuous pixel values. Thus, we simulate a de-quantisation of KMNIST images by adding a noise jitter term of variance approximately matching that of a uniform distribution over the quantisation step \cite{Hoogeboom2020dequantisation}.

Table \ref{tab:kmnist_test_log_lik} shows the test log-likelihood for all the methods and experimental settings under consideration. The peak signal-to-noise ratio (PSNR) and Structural Similarity Index (SSIM) of posterior mean reconstructions are given in \cref{tab:kmnist_image_PSNR}.
All methods show similar PSNR with the standard DIP (with TV regularisation) obtaining better PSNR in the very ill-posed setting ($5$ angles) and MCDO obtaining marginally better reconstruction in all others. Despite this, the linearised DIP outperforms all baselines in terms of test log-likelihood in all settings.
Furthermore, since the DIP provides 4dB higher PSNR reconstructions than the non-DIP based priors (cf. \cref{tab:HMC}), linearised DIP handily obtains a better test log-likelihood than these more-traditional methods.

\begin{table}[h!]
\caption{Mean and std-err of test log-likelihood computed over 50 KMNIST test images for all methods under consideration.}
\resizebox{\columnwidth}{!}{
\begin{tabular}{lcccc}
$\eta$ (5\%)&\hspace{-1cm}$\#$angles:\quad $5$ & $10$ & $20$ & $30$\\
\hline
DIP ($\sigma^2_y$ = 1) & 0.68 $\pm$ 0.14 & 1.57 $\pm$ 0.02 & 1.85 $\pm$ 0.02 & 2.02 $\pm$ 0.02\\
DIP-MCDO & 0.74 $\pm$ 0.13 & 1.60 $\pm$ 0.02 & 1.87 $\pm$ 0.02 & 2.05 $\pm$ 0.02\\
lin.-DIP (MLL) & \textbf{1.90} $\pm$ \textbf{0.14} & \textbf{2.57} $\pm$ \textbf{0.09} & \textbf{2.94} $\pm$ \textbf{0.10} & \textbf{3.09} $\pm$ \textbf{0.12}\\
lin.-DIP (TV-MAP) & \textbf{1.88} $\pm$ \textbf{0.15} & \textbf{2.59} $\pm$ \textbf{0.10} & \textbf{2.96} $\pm$ \textbf{0.10} & \textbf{3.11} $\pm$ \textbf{0.12}\\
\hline
\end{tabular}
}
\\[0.75em]
\resizebox{\columnwidth}{!}{
\begin{tabular}{lcccc}
$\eta$ (10\%)&\hspace{-1cm}$\#$angles:\quad $5$ & $10$ & $20$ & $30$\\
\hline
DIP ($\sigma^2_y$ = 1) & $0.27 \pm 0.17$ & $1.31 \pm 0.04$ & $1.62 \pm 0.03$ & $1.76 \pm 0.04$\\
DIP-MCDO & 0.42 $\pm$ 0.14 & 1.39 $\pm$ 0.04 & 1.70 $\pm$ 0.03 & 1.85 $\pm$ 0.04\\
lin.-DIP (MLL) & \textbf{1.63} $\pm$ \textbf{0.08} & \textbf{2.11} $\pm$ \textbf{0.07} & \textbf{2.43}  $\pm$ \textbf{0.07} & \textbf{2.59} $\pm$ \textbf{0.08}\\
lin.-DIP (TV-MAP) & \textbf{1.63} $\pm$ \textbf{0.09} & \textbf{2.13} $\pm$ \textbf{0.07} & \textbf{2.45} $\pm$ \textbf{0.08} & \textbf{2.61} $\pm$ \textbf{0.08}\\
\hline
\end{tabular}
}
\label{tab:kmnist_test_log_lik}
\end{table}

\begin{table}[h!]
\caption{PSNR [dB] / SSIM of the reconstruction posterior mean, averaged over 50 KMNIST test images for all inference methods.}
\resizebox{\columnwidth}{!}{%
\begin{tabular}{lcccc}
$\eta$ (5\%)&\hspace{-.25cm}$\#$angles:\quad $5$ & $10$ & $20$ & $30$\\
\hline
DIP &  \textbf{21.42}/ \textbf{0.890} & 27.92/ \textbf{0.977} & 31.21/ \textbf{0.988} & 32.93/ \textbf{0.991}\\
DIP-MCDO & 20.95/0.882 &  \textbf{28.26}/ \textbf{0.977} &  \textbf{31.65}/0.986 &  \textbf{33.45}/0.990\\
\hline
\end{tabular}%
}
\\[0.75em]
\resizebox{\columnwidth}{!}{%
\begin{tabular}{lcccc}
$\eta$ (10\%)&\hspace{-.25cm}$\#$angles:\quad $5$ & $10$ & $20$ & $30$\\
\hline
DIP &  \textbf{19.46}/ \textbf{0.846} & 24.56/ \textbf{0.956} & 27.27/ \textbf{0.974} & 28.57/ \textbf{0.980}\\
DIP-MCDO & 18.91/0.830 & \textbf{24.76}/0.953 & \textbf{27.72}/0.972 &  \textbf{29.09}/0.978\\
\hline
\end{tabular}%
}%
\label{tab:kmnist_image_PSNR}
\end{table}

\begin{table}[h!]
\caption{Evaluation of our approximate covariance estimation methods in terms of test log-likelihood over 10 KMNIST test images considering the 20 angle ($d_{y}=820$) setting and using lin.-DIP (MLL).}
\resizebox{\columnwidth}{!}{%
\begin{tabular}{lcccc}
$\eta$ (\%)& \shortstack{exact\\cov. \cref{eq:posterior_predictive}} & \shortstack{sampled\\cov. \cref{eq:matheron}} & \shortstack{sampled\\cov. ($\tilde \mJ$) \cref{eq:matheron}} & \shortstack{sampled cov.\\($\tilde \mJ$ $\&$ PCG) \cref{eq:matheron}}\\
\hline
$5$ & $2.80 \pm 0.06$ & $2.80 \pm 0.06$ & $2.68 \pm 0.09$ & $2.62 \pm 0.09$ \\
$10$ & $2.26 \pm 0.06$ & $2.26 \pm 0.06$ & $2.21 \pm 0.06$ & $2.22 \pm 0.06$ \\
\hline
\end{tabular}%
}
\label{tab:valid_sampl_approx}
\vspace{-1.25em}
\end{table}

\begin{figure*}[t!]
    \centering
    \includegraphics[width=\textwidth]{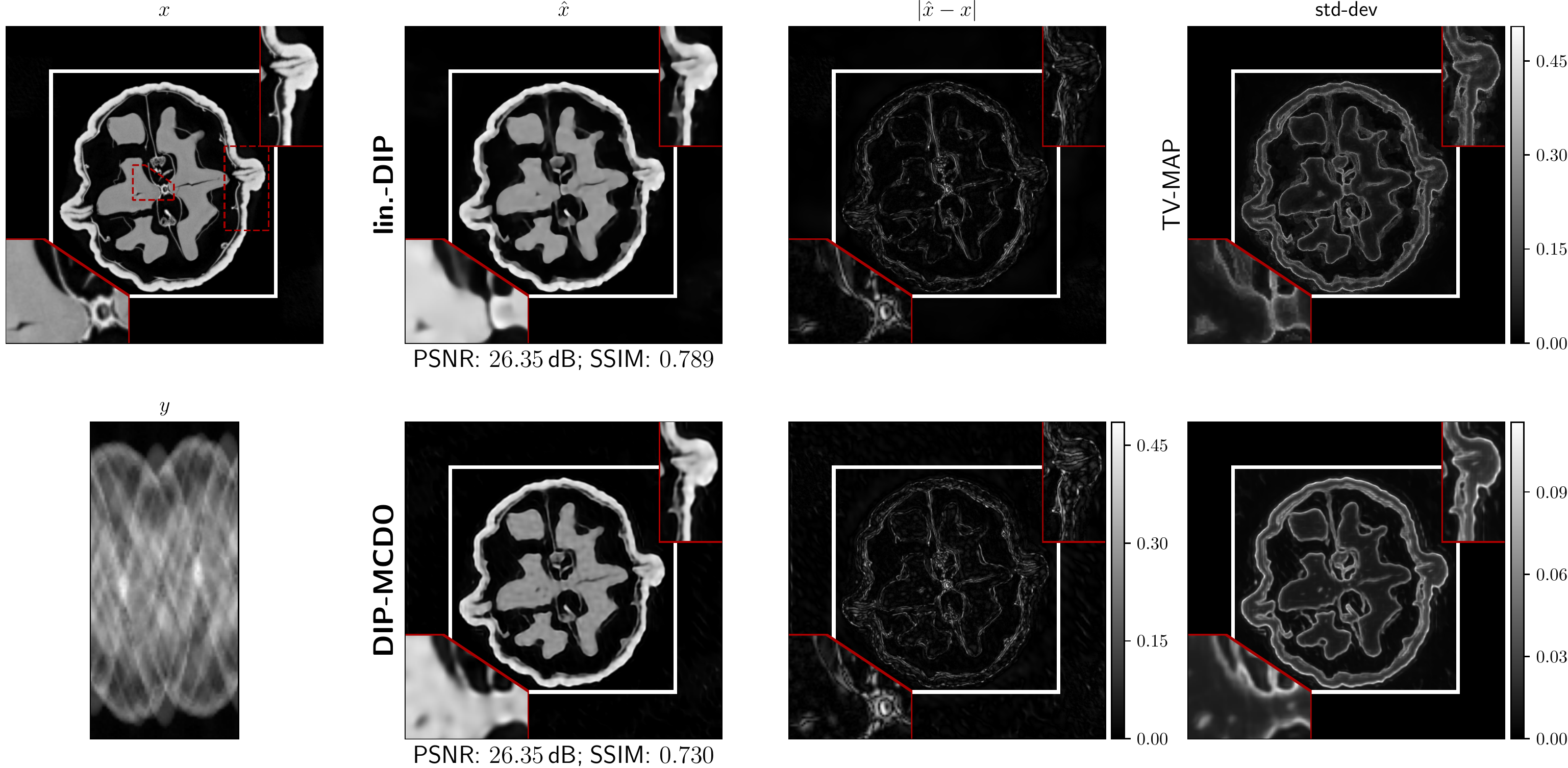}
    \vspace{-1.25em}
    \caption{
    The reconstruction of a Walnut using lin.-DIP and DIP-MCDO along with their respective uncertainty estimates.
    }
    \label{fig:main_walnut}
\end{figure*}

\Cref{fig:main_kmnist} shows an exemplary character recovered from a simulated observation $\data$ (using 20 angles and $5\%$ noise) with both linearised DIP and DIP-MCDO along with their associated uncertainty maps and calibration plots. DIP-MDCO systematically underestimates uncertainty for pixels on which the error is large, explaining its poor test log-likelihood.
The pixel-wise standard deviation provided by linearised DIP (TV-MAP) better correlates with the reconstruction error.

\subsubsection{Evaluating the fidelity of sample-based predictive covariance matrix estimation }

We evaluate the accuracy of the sampling, conjugate gradient and low rank based approximations to constructing the predictive covariance $\Sigma_{x|y}$ discussed in \cref{sec:computations}. As a reference, we compute the exact predictive covariance as in \cref{eq:posterior_predictive}, which is tractable for KMNIST. Our approximate methods use \cref{eq:matheron} to draw zero mean samples and use these to estimate $\hat \Sigma_{x|y} = 1/k \sum^{k}_{j=1} x_k x_{k}^\top$.
We construct the exact covariance matrix, forgoing patch-based approximations and stabilised estimators as to isolate the effect of using different approaches to sampling.
\Cref{tab:valid_sampl_approx} shows that estimating the covariance matrix using exact samples provides no decrease in performance relative to using the exact matrix. Using a low-rank approximation to the Jacobian matrix together with computing linear solves with PCG loses at most 0.32 nats in test log-likelihood with respect to the exact computation, but results in almost an order of magnitude speedup at prediction time.

\begin{figure}[t]
    \centering
    \includegraphics[width=0.9\columnwidth]{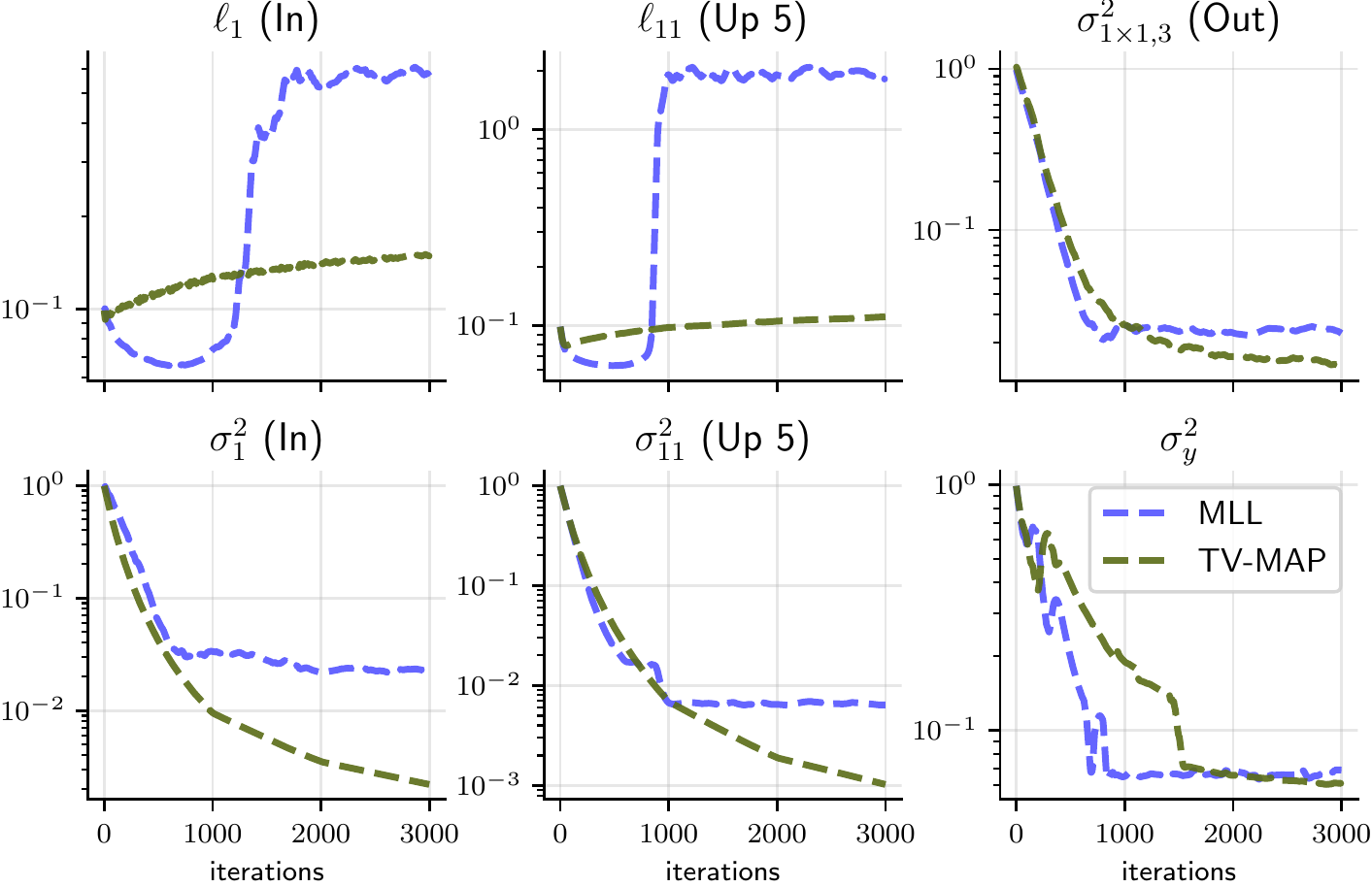}
    \caption{Optimisation trajectory for hyperparameters of the U-net's first and last $3\times3$ convolutions $(\ell_{1}, \sigma^2_{1}, \ell_{11},  \sigma^2_{11})$, last $1\times1$ convolution ($\sigma^2_{1\times1, 3}$) and noise variance $\sigma_y^{2}$, via MLL and Type-II MAP for the Walnut data.}
    \vspace{-1.25em}
    \label{fig:mrglik_opt_hyper-main}
\end{figure}

\subsection{Linearised DIP for high-resolution CT}\label{subsec:walnut-exps}

We now demonstrate our approach on real-measured cone-beam $\mu$CT data obtained by scanning walnuts, and released by \cite{der_sarkissian2019walnuts_data}.
We reconstruct a $501\times 501\,\text{px}^2$ slice ($d_x=251 \times 10^3$) from the first walnut of the dataset using a sparse subset of measurements taken from $60$ angles and $128$ detector rows ($d_y=7680$). Here, $\Sigma_{xx}$ is too large to store in memory and $\Sigma_{\data \data}$ is too expensive to assemble repeatedly. Furthermore, we use the deep architecture shown in \cref{fig:unet-diagram} which contains approximately 3 million parameters.
We thus resort to the approximate computations described in \cref{sec:computations}. Since the Walnut data is not quantised, jitter correction is not needed.

\Cref{fig:mrglik_opt_hyper-main} shows how Type-II MAP hyperparameters optimisation drives  $\sigma^{2}$ to smaller values, compared to MLL. 
This restricts the linearised DIP prior, and thus the induced posterior, to functions that are smooth in a TV sense, leading to smaller error-bars c.f.~\cref{fig:main_walnut_uq}.
During MLL and Type-II MAP optimisation, we observe that many layers' prior variance goes to $\approx 0$.  This phenomenon is known as ``automatic relevance determination'' \cite{Mackay1996BAYESIANNM,Tipping:2001}, and simplifies our linearised network, preventing uncertainty overestimation. We did not observe this effect when working with KMNIST images and smaller networks.
We display the MLL and MAP optimisation profiles for the active layers (i.e. layers with high $\sigma_d^2$) in \cref{fig:mrglik_opt_hyper-main}.
As the optimisation of \cref{eq:type2MAP} progresses, $\ell_1, \ell_{11}$ fall into basins of new minima corresponding to larger lengthscales. This results in more correlated dimensions in the prior, further simplifying~the~model.

\begin{table}[h!]
\centering
\caption{Test log-likelihood, PSNR and structural similarity (SSIM) on the Walnut. Comparing lin.-DIP against DIP-MCDO.}
\resizebox{\columnwidth}{!}{
\begin{tabular}{lccccc}
& \shortstack{$1\times 1$} & \shortstack{$2\times 2$} & \shortstack{$10\times 10$} & PSNR [dB] & SSIM\\
\hline
DIP-MCDO & 0.03 & 1.68 & 2.47 & \textbf{26.35} & 0.730\\
lin.-DIP (MLL) & 2.09 & 2.25 & 2.43 & \textbf{26.35} & \textbf{0.789}\\
lin.-DIP (MLL, $\tilde \mJ$ $\&$ PCG) & 1.88 & 2.05 & 2.24 & $-$ & $-$\\
lin.-DIP (TV-MAP) & \textbf{2.21} & \textbf{2.40} & \textbf{2.60} & $-$ & $-$\\
lin.-DIP (TV-MAP, $\tilde \mJ$ $\&$ PCG) & \textbf{2.24} & \textbf{2.46} & \textbf{2.65}  & $-$ & $-$\\
\hline
\end{tabular}
}
\label{tab:walnut_log_lik_and_metrics}
\end{table}

\begin{figure}
    \centering
    \vspace{-2.25em}
    \includegraphics[width=0.95\columnwidth]{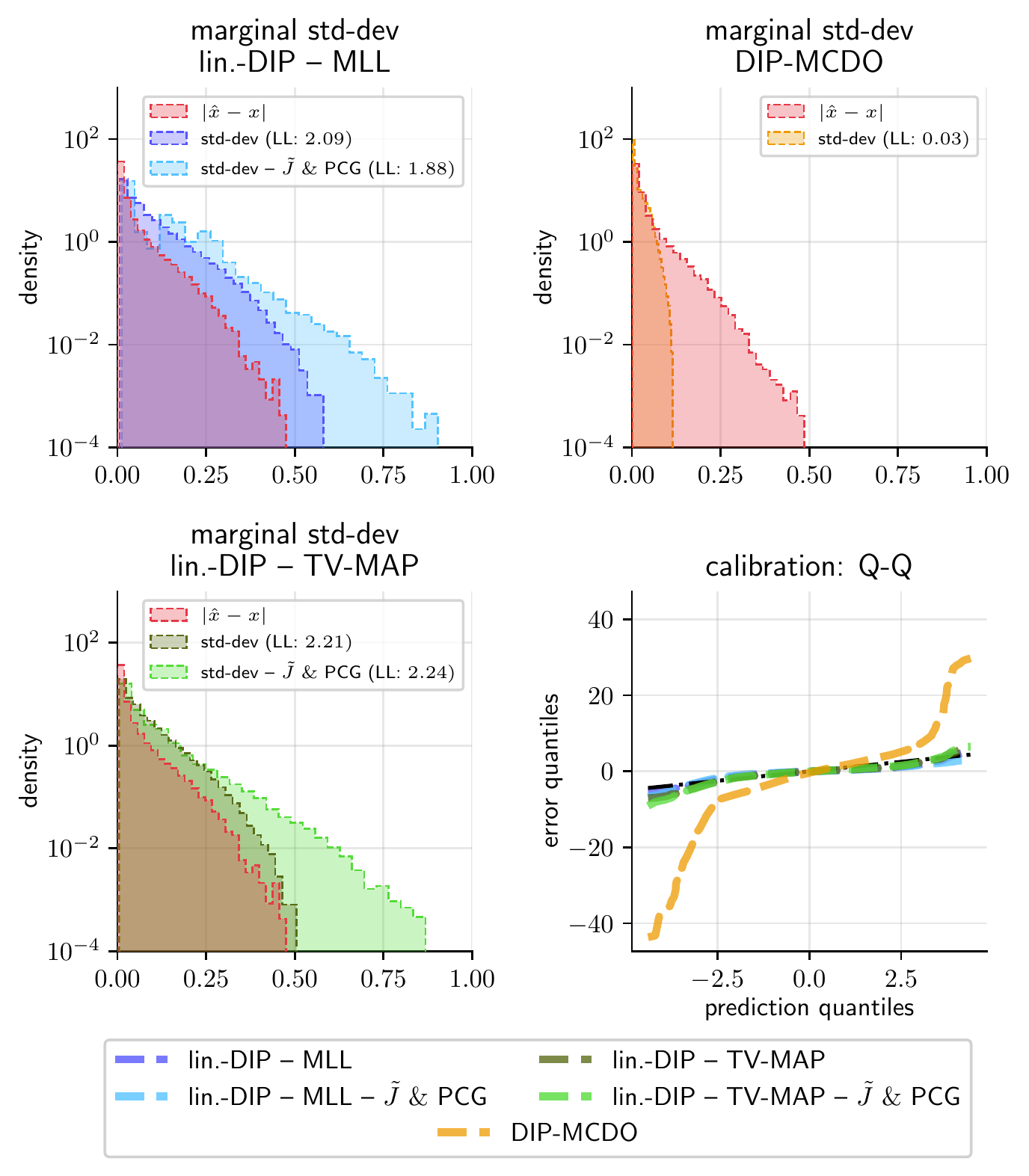}
   \caption{%
    The comparison of uncertainty calibration: the pixel-wise reconstruction absolute error $|\hat x - x|$ overlaps with the uncertainties provided by the lin.-DIP. DIP-MCDO, instead, severely underestimates uncertainty. The scale of the pixel-wise standard deviation (std-dev) obtained including the TV-PredCP matches the absolute error more closely than when the hyperparameters are optimised without. Using $\tilde J$ \& PCG results in overestimating uncertainty in the tails. LL stands for test log-likelihood.%
    }
    \vspace{-1.75em}
    \label{fig:main_walnut_uq}
\end{figure}

In \cref{tab:walnut_log_lik_and_metrics}, we report test log-likelihood computed using a Gaussian predictive distribution with covariance blocks of sizes  $1\times 1 $, $2\times 2$ and $10\times 10$ pixels, computed as described in \cref{apd:add_setup_walnut}. Mean reconstruction metrics are also reported.
Density estimation operations, described in \cref{sec:post_sampling}, are conducted in double precision (64 bit floating point) as we found single precision led to numerical instability in the assembly of $\Sigma_{\data \data}$, and also in the estimation of off-diagonal covariance terms for larger patches.
\Cref{fig:main_walnut} displays reconstructed images, uncertainty maps and calibration plots. In this more challenging tomographic reconstruction task, DIP-MCDO performs poorly relative to the standard regularised DIP formulation \cref{eq:DIP_MAP-obj2} in terms of reconstruction PSNR. Furthermore, the DIP-MCDO uncertainty map is blurred across large sections of the image, placing large uncertainty in well-reconstructed regions and vice-versa. In contrast, the uncertainty map provided by linearised DIP is fine-grained, concentrating on regions of increased reconstruction error. 
Quantitatively, linearised DIP provides over 2.06 nats per pixel improvement in terms of test log-likelihood and more calibrated uncertainty estimates, as reflected in the Q-Q plot. Furthermore, the use of our TV-PredCP based prior for MAP optimisation yields a 0.12 nat per pixel improvement over the MLL approach.

\section{Conclusion}\label{sec:conclusion}

We have proposed a probabilistic formulation of the deep image prior (DIP) that utilises the linearisation around the mode of the network weights and a Gaussian-linear hierarchical prior on the network parameters mimicking the total variation prior (constructed via predictive complexity prior). The approach yields well-calibrated uncertainty estimates on tomographic reconstruction tasks based on simulated observations as well as on real-measured $\mu$CT data.
The empirical results suggest that the DIP and the TV regulariser provide good inductive biases for both high-quality reconstructions and well-calibrated uncertainty estimates. The proposed method is shown to provide by far more calibrated uncertainty estimates than existing approaches to uncertainty estimation in DIP, like MC dropout.

\section*{Acknowledgements}
The authors would like to thank Marine Schimel, Alexander Terenin, Eric Nalisnick, Erik Daxberger and James Allingham for fruitful discussions. R.B. acknowledges support from the i4health PhD studentship (UK EPSRC EP/S021930/1), and from The Alan Turing Institute (UK EPSRC EP/N510129/1). J.L. was funded by the German Research Foundation (DFG; GRK 2224/1) and by the Federal Ministry of Education and Research via the DELETO project (BMBF, project number 05M20LBB). The work of BJ is partially supported by UK EPSRC grants EP/T000864/1 and EP/V026259/1. JMHL acknowledges support from a Turing AI Fellowship EP/V023756/1 and
an EPSRC Prosperity Partnership EP/T005386/1. JA acknowledges support from Microsoft Research, through its PhD Scholarship Programme, and from the EPSRC. This work has been performed using resources provided by the Cambridge Tier-2 system operated by the University of Cambridge Research Computing Service (http://www.hpc.cam.ac.uk) funded by EPSRC Tier-2 capital grant EP/T022159/1.

\bibliographystyle{acm}
\bibliography{bibliography}
\vspace{-4em}
\begin{IEEEbiography}[{\includegraphics[width=1in,height=1.25in,clip,keepaspectratio]{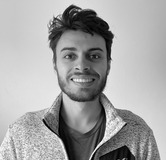}}]{Javier~Antorán} J. Antorán is a PhD student in the Machine Learning group within the Computational and Biological Learning Lab at the University of Cambridge, UK. Javier received his MPhil in Machine Learning (2019) from the University of Cambridge and his B.S. degree in Telecommunications engineering (2018) from the University of Zaragoza, Spain. Javier's research interests include probabilistic reasoning with neural networks, Gaussian processes and information theory.
\end{IEEEbiography}
\vspace{-5em}
\begin{IEEEbiography}[{\includegraphics[width=1in,height=1.25in,clip,keepaspectratio]{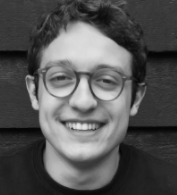}}]{Riccardo Barbano}
R. Barbano is a PhD student in the i4Health CDT at CMIC, University College London supervised by Professor Bangti Jin and Professor Simon Arridge. Previously, he obtained an MRes in Medical Imaging at University College London (2020), an MPhil in Machine Learning and Machine Intelligence at the University of Cambridge (2019) and an MEng in Engineering at Imperial College London (2018).
\end{IEEEbiography}
\vspace{-5em}
\begin{IEEEbiography}[{\includegraphics[width=1in,height=1in,clip]{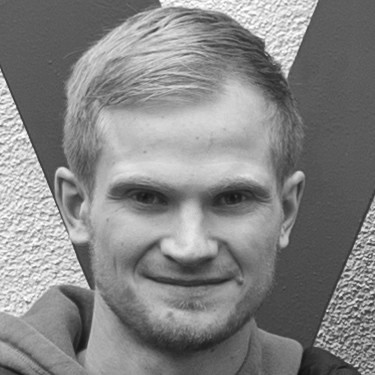}}]{Johannes Leuschner} J. Leuschner is a PhD student with the Research Training Group $\pi^3$ at the Center for Industrial Mathematics, University of Bremen, Germany, supervised by Professor Peter Maass. He received his MSc in Industrial Mathematics from the University of Bremen (2019). Johannes' research interests include deep learning methods for computed tomography.
\end{IEEEbiography}
\vspace{-5em}
\begin{IEEEbiography}[{\includegraphics[width=1in,height=1.25in,clip,keepaspectratio]{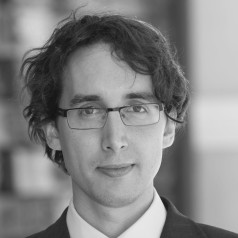}}]{José~Miguel~Hernández-Lobato}
is Professor of Machine Learning at the Engineering Department from University of Cambridge, UK. Before this, he was a postdoctoral fellow at Harvard University (2014--2016) and a postdoctoral research associate at University of Cambridge (2011--2014). He completed his Ph.D. (2010) and M.Phil. (2006) in Computer Science at Universidad Autónoma de Madrid (2010). He also holds a B.Sc. in Computer Science from this institution (2004). His research is on probabilistic machine learning and its applications to real-world problems.
\end{IEEEbiography}
\vspace{-5em}
\begin{IEEEbiography}[{\includegraphics[width=1in,height=1.25in,clip,keepaspectratio]{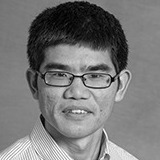}}]{Bangti Jin}
B. Jin received a PhD in Mathematics from the Chinese University of Hong Kong, Hong Kong in 2008. Previously, he was Lecturer and Reader, and Professor at Department of Computer Science, University College London (2014-2022), an assistant professor of Mathematics at the University of California, Riverside (2013–2014), a visiting assistant professor at Texas A\&M University (2010–2013), an Alexandre von Humboldt Postdoctoral Researcher at University of Bremen (2009–2010). Currently he is Professor of Mathematics at the Chinese University of Hong Kong.
\end{IEEEbiography}

\onecolumn
\renewcommand{\Alph}

\section*{Supplementary material}
\setcounter{section}{0}

\section{Designing total variation priors}

To develop a probabilistic DIP, we describe first how to design a tractable TV prior for computational tomography.
To this end, we reinterpret the TV regulariser \cref{eq:TV_equation} as a prior over images, favouring those with low $L^1$ norm gradients
\begin{gather}\label{eq:generic_tv_prior_apd}
   p(x) = Z_{\lambda}^{-1} \exp({-\lambda\TV(x)}),
\end{gather}
where $Z_{\lambda} = \int \exp(-\lambda\TV(x))\, {\rm d} x$. 
This prior is intractable because $Z_{\lambda}$ does not admit a closed form; thus approximations are necessary.
We now explore alternatives without this limitation.

\subsection{Further discussion on the TV regulariser as a prior}\label{apd:TV}

It is tempting to think that we do not need the PredCP machinery in \cref{subseq:tv_for_NN} to translate the TV regulariser into the parameter space. Indeed, the Laplace method simply involves a quadratic approximation around a mode of the log posterior, without placing any requirements on the prior used to induce said posterior. Hence, we can decompose the Hessian of the log posterior $\log p(\params|\data)$ into the contributions from the likelihood and the prior as
\begin{gather*}
    \dfrac{\partial}{\partial \theta^2} \left(\log p(\data | \op \net)  + \log p(\net) \right)|_{\params = \hat\params}
\end{gather*}
and quickly realise that the log of the anisotropic $\TV$ prior chosen to be $p(x) \propto \exp(-\lambda \TV(x))$ as in \cref{eq:generic_tv_prior_apd} is only once differentiable. Ignoring the origin (where the absolute value function is non-differentiable), we obtain:
\begin{gather*}
    \dfrac{\partial}{\partial \theta^2} \log p(\net)|_{\params = \hat\params} \propto - \dfrac{\partial}{\partial \theta^2} \TV(\net)|_{\params = \hat\params} = 0.
\end{gather*}
Thus, a naive application of the Laplace approximation would eliminate the effect of the prior, leaving the posterior ill defined. In practice, one may smooth the non-smooth region around the origin, but the amount of smoothing can significantly influence the behaviour of the Hessian approximation.

\subsection{Further discussion on inducing TV-smoothness with Gaussian priors}

A standard alternative to enforce local smoothness in an image is to adopt a Gaussian prior $p(x) = \mathcal{N}(x; \mu, \Sigma_{xx})$ with covariance $\Sigma_{xx} \in \BR^{d_{x} \times d_{x}}$ given by
\begin{gather*}\label{eq:Gaussian_over_x_apd}
     [\Sigma_{xx}]_{ij,i'j'} = \sigma^{2}\exp\left({\frac{-\dist(i-i', j-j')}{\ell}}\right),
\end{gather*}
where $i, j$ index the spatial locations of pixels of $x$, as in \cref{eq:TV_equation}, and $\dist(a, b) = \sqrt{a^{2} + b^{2}}$ denotes the Euclidean vector norm. \Cref{eq:Gaussian_over_x_apd} is also known as the Matern-$\nicefrac{1}{2}$ kernel and matches the covariance of Brownian motion \cite{maternmultidim}. 
The hyperparameter $\sigma^{2} \in \BR^{+}$ informs the pixel amplitude while the lengthscale parameter $\ell \in \BR^{+}$ determines the correlation strength between nearby pixels. 
The TV in \cref{eq:TV_equation} only depends on pixel pairs separated by one pixel ($\dist = 1$), allowing analytical computation of the expected TV associated with the Gaussian prior
\begin{gather}
    \kappa := \BE_{x\sim\mathcal{N}(x; \mu, \Sigma_{xx})}[\TV({x})] = c \sqrt{\sigma^{2} (1 - \rho)},
\end{gather}
with the correlation coefficient $\rho = \exp(-\ell^{-1}) \in (0, 1)$ and $c = \nicefrac{4  \sqrt{d_{x}}  (\sqrt{d_{x}}-1)}{\sqrt{\pi}}$ for square images. See Appendix \cref{apd:exact_ETV} for derivations.
Increasing $\ell$ (for a fixed $\sigma^{2}$) favours $x$ with low TV on average, resulting in smoother images. The prior $\mathcal{N}(x; \mu, \Sigma_{xx})$ is conjugate to the likelihood implied by the least-square fidelity $\mathcal{N}(\data; \op x, \sigma_{y}^{2} \mI)$, leading to a closed form posterior predictive distribution and marginal likelihood objective with costs $\mcO({d_{y}^{3}})$ and $\mcO({d_{y}^{2}d_{x}})$, respectively. Their expressions match those provided for the DIP in \cref{sec:probabilistic_model} of the main text.

\subsection{Derivation of the identity \cref{eq:exact_ETV}}\label{apd:exact_ETV}

The identity follows from the following result (appendix, \cite{McGrawWong:1994}). 
The short proof is recalled for the convenience of the reader.
\begin{lemma}\label{lem:absolute-diff}
Let $X$ and $Y$ be normal random variables with mean $\mu$, variance $\sigma^2$ and correlation coefficient $\rho$. Let $Z=|X-Y|$. Then 
\begin{equation*}
    \mathbb{E}[Z]=\frac{2}{\sqrt{\pi}}\sqrt{\sigma^2(1-\rho)}.
\end{equation*}
\end{lemma}
\begin{proof}
Clearly, $X-Y$ follows a Gaussian distribution with mean 0 and variance $2\sigma^2(1-\rho)$. Then the random variable 
\begin{equation*}
W=\frac{Z^2}{2\sigma^2(1-\rho)}=\Big(\frac{X-Y}{\sqrt{2\sigma^2(1-\rho)}}\Big)^2
\end{equation*}
follows $\chi_1^2$ distribution. Then 
\begin{equation*}
    \mathbb{E}[\sqrt{W}]=\int_0^\infty W^\frac12 \frac{1}{\Gamma(\frac12)\sqrt{2}}W^{\frac{1}{2}-1}e^{-\frac{W}{2}}\mathrm{d}W = \frac{\sqrt{2}}{\Gamma(\frac12)}=\frac{\sqrt{2}}{\sqrt{\pi}},
\end{equation*}
where $\Gamma(z)$ denotes the Euler's Gamma function, with $\Gamma(\frac12)=\sqrt{\pi}$. Then it follows that
\begin{equation*}
    \mathbb{E}[Z]=\sqrt{2\sigma^2(1-\rho)}\mathbb{E}[\sqrt{W}]=\frac{2}{\sqrt{\pi}}\sqrt{\sigma^2(1-\rho)}.
\end{equation*}
This shows the assertion in the lemma.
\end{proof}

Now by the marginalisation property of multivariate Gaussians, any two neighbouring pixels of $x$ for $x\sim\mathcal{N}(x; \mu, \Sigma_{xx})$ satisfy the conditions of Lemma \ref{lem:absolute-diff}, with $\rho=\exp(-\ell^{-1}) \in (0, 1)$. Thus Lemma \ref{lem:absolute-diff} and the trivial fact $d_x = h \times w$ imply 
\begin{equation*}
    \kappa_d= \BE_{\mathcal{N}(x;\mu, \Sigma_{xx})}[\TV({x})]=\frac{2[2hw-h-w]}{\sqrt{\pi}}\sqrt{\sigma^2(1-\rho)}.
\end{equation*}
In particular, for a square image, $h=w=\sqrt{d_x}$, we obtain the desired identity in \cref{eq:exact_ETV}.

\section{Derivation of the Bayes deep image prior}\label{apd:inference}

\subsection{Posterior predictive covariance}\label{apd:inference_on_11}

We provide an alternative derivation of the posterior predictive covariance of the linearised DIP by reasoning in the parameter space. 
First we have linearised the neural network $\net$, turning it into a Bayesian basis function linear model \cite{Khan19approximate}. The probabilistic model in \cref{eq:Model_weight_space} is thus:
\begin{gather*}
    \data|\params \sim \mathcal{N}(\data; \op h(\params), \sigma_{y}^{2} \mI), \quad \params|\ell \sim \mathcal{N}\left(\params; 0, \Sigma_{\params\params}\right),
\end{gather*}
and the linearised Laplace approximate posterior distribution over weights is given by \cite{Immer2021improving}
\begin{gather}
    p(\params | \data) \approx \mathcal{N}(\params; \hat\params, \Sigma_{\params|\data}) \spaced{with} \Sigma_{\params|\data} = \left( \sigma_{y}^{-2} \mJ^{\top} \op^{\top} \op  \mJ  + \Sigma_{\params\params}^{-1}\right)^{-1}.
\end{gather}
In this work we exploit the equivalence between basis function linear models and Gaussian Processes (GP), and perform inference using the dual GP formulation. 
This is advantageous due to its lower computational cost when $d_{\theta} >> d_{y}$, which is common in tomographic reconstruction.

We switch to the dual formulation using the SMW matrix inversion identity, we have
\begin{equation}
   \Sigma_{\params|\data} = \left( {\sigma_{y}^{-2}} \mJ^{\top} \op^{\top} \op  \mJ   + \Sigma_{\params\params}^{-1}  \right)^{-1} = \Sigma_{\params\params} - \Sigma_{\params\params}\mJ^{\top} \op^{\top}(\sigma_{y}^{2} \mI + \op \mJ \Sigma_{\params\params}^{-1} \mJ^{\top} \op^{\top})^{-1}\op \mJ \Sigma_{\params\params}
\end{equation}
The predictive distribution over images can be built by marginalising the NN parameters in the conditional likelihood $ p(x|\data) = \int p(x | \params)p(\params | \data)\,{\rm d}\params$. Because $h(\cdot)$ is a deterministic function, we have that $p(x | \params) = \delta(x - h(\params))$ and thus
\begin{gather*}
    \int p(x | \params)p(\params | \data)\,{\rm d}\params = \int \delta(x - h(\params)) \mathcal{N}(\params; \hat\params, \Sigma_{\params|\data}) \,{\rm d}\params = \mathcal{N}(x; \hat x, \mJ \Sigma_{\params|\data} \mJ^{\top}).
\end{gather*}
Note that this assumes $\hat \theta$ to be a mode of the DIP training loss \cref{eq:DIP_MAP-obj2}. In practise, this will not be satisfied and thus the posterior mean of the linear model $\hat \theta_h$, which is given as the minima of the linear model's loss introduced in \cref{subsec:marginal_likelihood}, will not match that of the NN, that is, $\hat \theta$. Using the linear model's exact mode is only necessary for the purpose of constructing the marginal likelihood objective \cite{antoran2021fixing,antoran2022Adapting} (see also \cref{apd:inference_on_12}). However, for the purpose of making predictions, assuming $\hat \theta$ to be the mode allows us to keep the DIP reconstruction $\hat x$ as the predictive mean.

\subsection{Laplace marginal likelihood and Type-II MAP in \cref{eq:type2MAP}}\label{apd:inference_on_12}

For the purpose of uncertainty estimation, we tune the hyperparameters of our linear model using the marginal likelihood of the conditional-on-$\ell$ Gaussian-linear model introduced in \cref{eq:linearised_model}. The posterior mode of the TV-regularised linearised model is given by $\hat \theta_h = \argmin_{\theta_{h}} \sigma_{y}^{-2}\|\op h(\theta_{h}) - y\| + \lambda \TV({h}(\params_h))$. However, we substitute the TV with a multivariate Gaussian surrogate $p(\theta | \ell)$. Now we derive the marginal log-likelihood (MLL) for the linearised model conditional on $\ell$ following \cite{antoran2022Adapting}. We start from Bayes rule
\begin{gather*}
   \log p(\theta | \data, \ell;  \sigma_y^{2}, \sigma^{2}) = \log p(\data | \theta; \sigma_y^{2} ) + \log p(\theta | \ell; \sigma^{2} )  - \log  \,p(\data |\ell; \sigma^{2}_{y}, \sigma^{2} ).
\end{gather*}
We now isolate the MLL $\log  \,p(\data |\ell; \sigma^{2}_{y}, \sigma^{2} )$ and evaluate at the linear model's posterior mode $\theta = \hat \theta_h$ and obtain 
\begin{gather}
    \log p(\data |\ell; \sigma^{2}_{y}, \sigma^{2} ) = \log p(\data | \theta{=}\hat \theta_h; \sigma_y^{2} ) + \log p(\theta{=}\hat\params_h | \ell; \sigma^2) - \log  p(\theta{=}\hat\params_h | \data, \ell;  \sigma_y^{2}, \sigma^{2}).
\end{gather}
The observation log-density $\log p(\data | \theta=\hat\params_h; \sigma_{y}^2)$ quantifies the quality of the model's fit to the data. It is given by
\begin{align*}
     \log p(\data | \theta=\hat\params_h; \sigma_{y}^2) &= - \frac{d_{y}}{2}\log (2\pi)-  \frac{1}{2} \log |\sigma^{2}_{y} \mI| - \frac{1}{2\sigma^{2}_{y}} \|\data - Ah(\hat \theta_h)\|^{2}_{2}.
 \end{align*}
However, since our predictive mode is given by the DIP reconstruction and not the linear model's reconstruction, we make a practical departure from the exact expression for the linear model's MLL and use 
$- \frac{d_{y}}{2}\log (2\pi)-  \frac{1}{2} \log |\sigma^{2}_{y} \mI| - \frac{1}{2\sigma^{2}_{y}} \|\data - \op\netmap\|^{2}_{2}$
as the data fit term instead. The weight-mode log prior density $ \log p(\theta{=}\hat\params_h | \ell, \sigma^2)$ is given by
\begin{align*}
    \log p(\theta{=}\hat\params_h | \ell, \sigma^2) &= - \frac{d_{\theta}}{2}\log (2\pi)-  \frac{1}{2} \log |\Sigma_{\params\params}| - \frac{1}{2}
    \hat \params_{h}^{\top} \Sigma^{-1}_{\params\params}\hat \params_{h}.
\end{align*}
Evaluating the Gaussian posterior log density over $\theta$ at its mode $\hat \theta_h$ cancels the exponent of the Gaussian and leaves us with just the normalising constant
\begin{gather*}
        \log  p(\theta{=}\hat\params_h | \data, \ell;  \sigma_y^{2}, \sigma^{2}) = -\frac{1}{2}\log |\Sigma_{\theta | y}| -\frac{d_{\theta}}{2} \log (2\pi)
\end{gather*}

By the matrix determinant lemma, the determinant $|\Sigma_{\theta|y}|$ is given by
\begin{gather}
   |\Sigma_{\theta|y}| = |\sigma_{y}^{-2}  \mJ^{\top} \op^{\top} \op \mJ   + \Sigma_{\params\params}^{-1}|^{-1} = | \op \mJ \Sigma_{\params\params} \mJ^{\top} \op^{\top} + \sigma_{y}^{2} \mI |^{-1} |\Sigma_{\params\params}|   |\sigma_{y}^{2} \mI|.
\end{gather}
Thus, the linearised Laplace marginal likelihood is given by 
\begin{align}
    \log p(\data | \ell; \sigma_{y}^{2}, \sigma^{2}) 
    =& -  \frac{1}{2} \log |\sigma^{2}_{y} \mI| - \frac{1}{2\sigma^{2}_{y}} \|\data - \op\netmap\|^{2}_{2} 
    -  \frac{1}{2} \log |\Sigma_{\params\params}| -  \frac{1}{2} \hat \params_{h}^{\top} \Sigma^{-1}_{\params\params} \hat \params_{h} \nonumber \\
    &-  \frac{1}{2} \log | \op \mJ \Sigma_{\params\params} \mJ^{\top} \op^{\top} + \sigma_{y}^{2} \mI | +  \frac{1}{2} \log |\Sigma_{\params\params}| +   \frac{1}{2} \log |\sigma_{y}^{2} \mI| + C\notag \\
    =& - \frac{1}{2 \sigma^{2}_{y} } ||\data - \op\netmap||^{2}_{2} - \frac{1}{2} \hat \params_{h}^{\top} \Sigma^{-1}_{\params\params}\hat \params_{h} - \frac{1}{2} \log | \op \mJ \Sigma_{\params\params} \mJ^{\top} \op^{\top} + \sigma_{y}^{2} \mI | + C\label{eq:derived_MLL}
\end{align}
where $C$ captures all terms constant in $(\sigma_{y}^{2}, \ell, \sigma^{2})$. Recall that $\Sigma_{\data \data} = \op \mJ \Sigma_{\params\params} \mJ^{\top} \op^{\top} + \sigma_{y}^{2} \mI$. 
Next we turn to the TV-PredCP prior over $\ell$
\begin{gather*}
    \log p(\ell; \sigma^{2}) = 
    - \sum_{d=1}^{D}\kappa_{d} + \log \left|\frac{\partial \kappa_{d}}{\partial \ell_{d}}\right|, \,\, \text{with} \,\,
    \kappa_{d} := \BE_{
    \mathcal{N}(\theta_d; \hat \theta_d, \Sigma_{\theta_d\theta_d})
    \prod^{D}_{i=1,i\neq d}\delta(\params_{i} - \hat \params_{i})}\left[{\lambda \TV}({h}(\params)) \right].
\end{gather*}
Hence we obtain the following Type-II maximum a posteriori (MAP)-style objective:
\begin{align*}
    &\log p(\data, \ell; \sigma^{2}_{y}, \sigma^{2}) \approx \log \mathcal{N}(\data; {0}, \Sigma_{\data \data}) + \log p(\ell; \sigma^{2}) \notag \\
    = & \frac{1}{2}\left(-\sigma^{-2}_{y}||\data - \op\netmap||_{2}^{2} - \hat \params_{h}^{\top} \Sigma^{-1}_{\params\params}\hat \params_{h} - \log |\Sigma_{\data \data}| \right)
    - \sum_{d=1}^{D} \kappa_{d} + \log \left|\frac{\partial \kappa_{d}}{\partial \ell_{d}}\right| +  C.
\end{align*}

\section{Additional details on our TV-PredCP}

\subsection{Correspondence to the formulation of \cite{nalisnick2021predicitve}}\label{app:predcp_relate_to_original}

The original formulation of the TV-PredCP \cite{nalisnick2021predicitve} defines a base model $q(x) = p(x|a=a_{0})$ and an extended model $p(x) = p(x|a = \tau)$. The (hyper)parameter $\tau$ determines how much the predictions of the two models vary. A divergence $\mathcal{D}(p(x|a=a_{0}) || p(x|a = \tau))$ is placed between the two distributions and a prior placed over the divergence. This divergence is mapped back to the parameter $\tau$ using the change of variables formula. To see how our approach \cref{eq:prior_ell} falls within this setup, take $p(x|a=\tau)$ to be $p(x) = \mathcal{N}(x; {\mu}, \Sigma_{xx}(\sigma^{2}, \ell))$, where the lengthscale $\ell$ takes the place of $\tau$.
The base model sets the lengthscale to be infinite, or equivalently the correlation coefficient $\rho$ to be 1,  $q(x) = \mathcal{N}(x; {\mu}, \Sigma_{xx}(\sigma_{x}^{2}, \infty))$. As a divergence, we choose $\mathcal{D}(p, q) = \BE_{p}[\text{TV}(x)] - \BE_{q}[\text{TV}(x)]$. We have defined our base model to be one in which all pixels are perfectly correlated and thus have the same value. This results in the expected TV for this distribution taking a value of 0.  We end up with our divergence simply matching the expected TV under the extended model $\BE_{\mathcal{N}(x; {\mu}, \Sigma_{xx})}[\text{TV}(x)]$. Even when an expected TV of 0 is not attainable for any value of $\ell$, as is the case when using the DIP \cref{eq:ell_prior_DIP}, there still exists a base model which will be constant with respect to our parameters of interest and can be safely ignored. 

\subsection{An upper bound on the expected TV}\label{app:predcp_upper_bound}

To ensure dimensionality preservation, we define our prior over $\ell$ in \cref{eq:ell_prior_DIP} as a product of TV-PredCP priors, one defined for every convolutional block in the CNN, indexed by $d$,
\begin{gather*}
    p({\ell}) = p(\ell_{1}) p(\ell_{2})\,...\,p(\ell_{D}) =\prod_{d=1}^{D}\pi(\kappa_{d})\left|\frac{\partial \kappa_{d}}{\partial \ell_{d}}\right|,
    \text{with }\kappa_{d} := \BE_{\mathcal{N}(\theta_d;\hat \theta_d, \Sigma_{\theta_d\theta_d})\prod^{D}_{i=1,i\neq d}\delta(\params_{i} - \hat \params_{i})}\left[{ \TV}({h}(\params)) \right].
\end{gather*}
This formulation differs from the expected TV introduced in \cref{eq:exact_ETV}, which does not discriminate by blocks $\kappa := \BE_{\mathcal{N}(\theta;\hat \theta, \Sigma_{\theta\theta})}\left[{ \TV}({h}(\params)) \right]$.
By the triangle inequality, $\sum_{d} \kappa_{d}$ is an upper bound on the expectation under the joint distribution 
\begin{gather*}
 \BE_{\mathcal{N}(\theta;\hat \theta, \Sigma_{\theta\theta})}\left[{\TV}({h}(\params)) \right] = \sum_{(i,j) \in \mathcal{S}} \BE_{\mathcal{N}(\theta;\hat \theta, \Sigma_{\theta\theta})}\left[|(\mJ_{i} \params -  \mJ_{j} \params) | \right] = \sum_{(i,j) \in \mathcal{S}} \BE_{\mathcal{N}(\theta;\hat \theta, \Sigma_{\theta\theta})}\left[| \sum_{d} (\mJ_{id} -  \mJ_{jd}) \params_{d}) | \right]\\
\leq \sum_{(i,j) \in \mathcal{S}} \sum_{d} \BE_{\mathcal{N}(\theta_d;\hat \theta_d, \Sigma_{\theta_d\theta_d})}\left[| (\mJ_{id} -  \mJ_{jd}) \params_{d} | \right] =  \sum_{d} \BE_{\mathcal{N}(\theta_d;\hat \theta_d, \Sigma_{\theta_d\theta_d}) \prod^{D}_{c=1,c\neq d}\delta(\params_{c} - \hat \theta_c)}\left[\sum_{(i,j) \in \mathcal{S}} | (\mJ_{i} -  \mJ_{j}) \params | \right] = \sum_{d} \kappa_{d},
\end{gather*}
where $\mathcal{S}$ is the set of all adjacent pixel pairs. 
Thus, the separable form of the TV prior as a regulariser for MAP optimisation ensures that the expectated TV under the joint distribution of parameters is also being regularised.

\subsection{Discussing monotonicity of the TV in the prior lengthscales}\label{apd:monotonicity}

In order to apply the change of variables formula in \cref{eq:ell_prior_DIP}, we require bijectivity between $\ell_{d}$ and $\kappa_{d}$. In the simplest setting, both variables are one-dimensional, making this constraint easier to satisfy. In fact, it suffices to show monotonicity between the two.
\begin{figure*}
    \centering
    \vspace{-0.1cm}
    \includegraphics[trim={0 0.85cm 0 0},clip, width=0.95\linewidth]{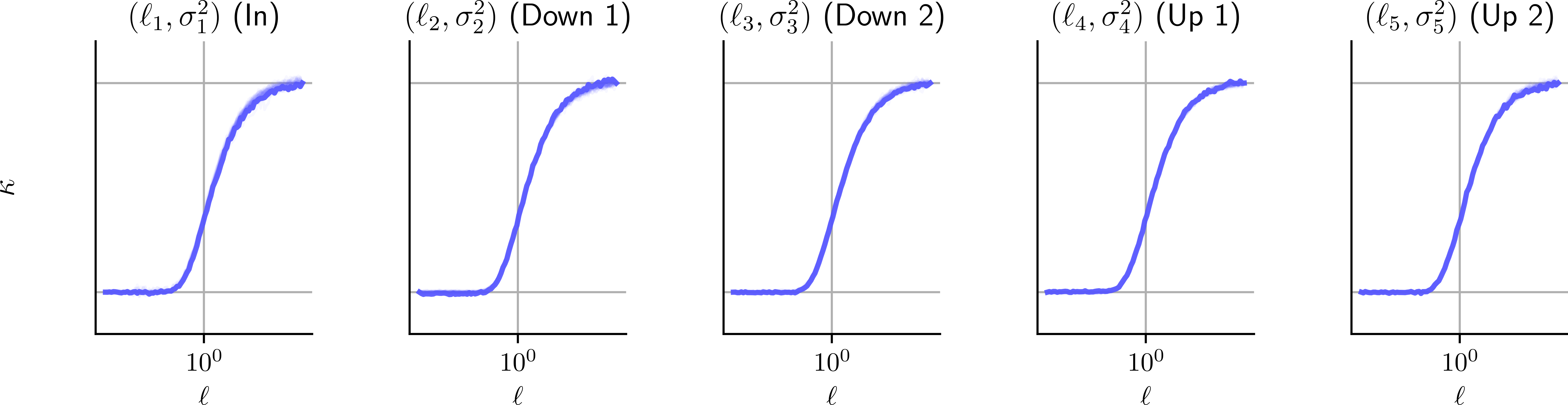}\\[0.5em]
    \includegraphics[width=0.95\linewidth, trim={0 0 0 0.45cm},clip]{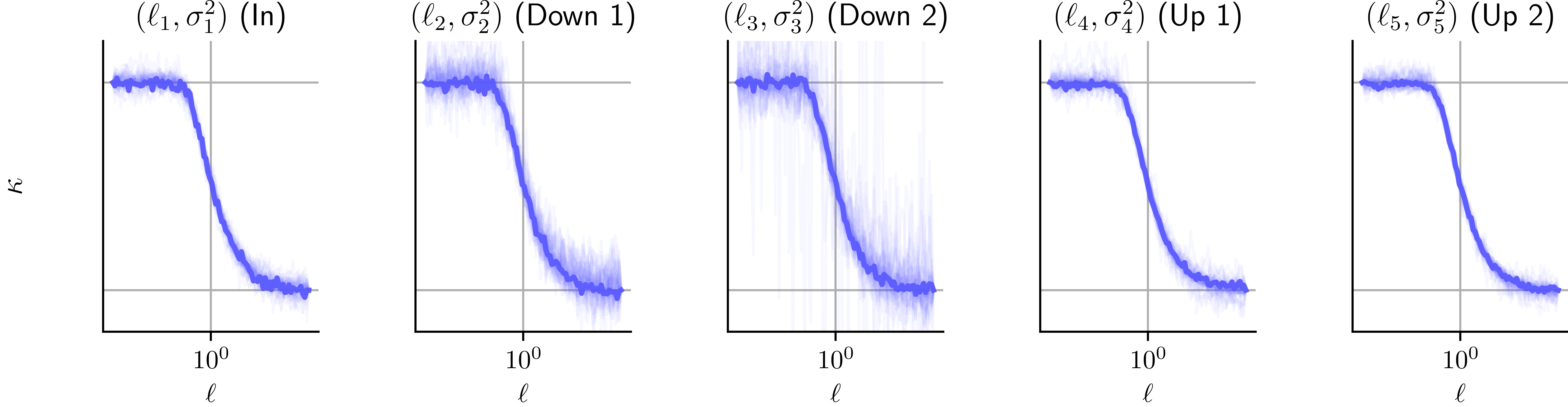}
    \vspace{-0.5cm}
    \caption{
    Experimental evidence of monotonicity computed over 50 KMNIST test images for the linearised network used in the KMNIST experiments. Horizontal axis represents lengthscale $\ell \in [0.01, 100]$. $\kappa$ is estimated with $10$k Monte Carlo samples. In the bottom row we fix the marginal variances of $J\Sigma_{\data\data}J^\top$ in image space to be 1. This allows us to observe the smoothing effect from $\ell$. We use the first and last value to normalise over different KMNIST sample. The monotonicity implies the desired invertibility of the mappings $\ell$ and $\kappa$. We draw 500 samples to estimate $k$.}
    \label{fig:monotonicity_linear_regime}
\end{figure*}
In practice, we use the linearised model in \cref{eq:linearised_model} for inference. In \cref{fig:monotonicity_linear_regime}, we show very compelling numerical evidence for the monotonicity. We observe that $\kappa$ increases in $\ell$ since large values for $\ell$ lead to an increased marginal variance $\sigma^2$ over images. After fixing the marginal variance to 1, the lengthscales have a monotonically decreasing relationship with the expected TV.
However, analytically studying the monotonicity remains delicate. 
We investigate the issue in the linear setting to she insights (which also matches our experimental setup):
\begin{equation}
    \kappa_{d} = \BE_{\mathcal{N}(\theta; \hat \theta, \Sigma_{\theta\theta})\prod^{D}_{j=1,j\neq d}\delta(\params_{j} - \hat \params_{j})}[ \TV({h}(\params))] 
    = \BE_{\mathcal{N}(\theta; \hat \theta, \Sigma_{\theta\theta})\prod^{D}_{j=1,j\neq d}\delta(\params_{j} - \hat \params_{j})}\Big[ \sum_{i} |{h}(\params)_{i} - {h}(\params)_{i+1} |\Big],\label{eq:k_d-ell}
\end{equation}
assuming that the output is a 1D signal so there is only one derivative to simplify the discussion.
First we derive the distribution of ${h}(\params)_{i} - {h}(\params)_{i+1}$. Note that ${h}(\params)$ can be written as $ {h}(\params) = {h}_0 + \mJ(\params - \hat\params)$, by slightly abusing the notation ${h}_0$ to denote the vectors constant with respect to $\ell_{d}$ and $i$ indices an entry of the vector $(\mJ \params) \in \BR^{d_x}$. Note that the constant vector ${h}_0$ depends on the choice of the based point $\params=0$ (or equally plausible $\params=\hat \params$),  but it does not play a role in ${\rm TV}({h}(\params))$, since it cancels out from the definition of ${\rm TV}({h}(\params))$. Then, we can rewrite it as an inner product between two vectors
\begin{gather*}
    {h}(\params)_{i} - {h}(\params)_{i+1}= (\mJ \params)_{i} - (\mJ \params)_{i+1} = (\mJ_{i} - \mJ_{i+1})\params_d = {v}_{i} \params_d,
\end{gather*}
where $\mJ_{i} \in \BR^{1\times d_{\theta_d}}$ denotes our NN's Jacobian for a single output pixel $i$ (i.e. the $i$th row of the Jacobian matrix $\mJ$, corresponding to the block parameters $\params_d\in \mathbb{R}^{1 \times d_{\theta_d}}$) and ${v}_{i} = \mJ_{i} - \mJ_{i+1}\in\mathbb{R}^{1\times d_{\theta_d}}$, $i=1,\ldots,d_x-1$.
Now, the block parameters $\params_{d}$ is distributed as 
\begin{equation*}
    \params_d \sim \mathcal{N}(\params_{d}; {0}, \Sigma_{\theta_d \theta_d}),
\end{equation*} 
in the expectation in \cref{eq:k_d-ell}, whereas the remaining parameters are fixed at the mode $\hat \params_j$, $j\neq d$, i.e. $\prod^{D}_{j=1,j\neq d}\delta(\params_{j} - \hat \params_{j}) $. Let ${V}_{d}\in \mathbb{R}^{ ( d_{x}-1)\times d_{\theta_d}}$ correspond to the stacking of the vectors ${v}_i\in \mathbb{R}^{1\times d_{\theta_d}}$, i.e. the Jacobian of the network output with respect to the weights in convolutional group $d$. Since the affine transformation of a Gaussian distribution remains Gaussian, ${V}_d\params_d$ is distributed according to
${V}_{d} \params_d \sim \mathcal{N}({0}, {V}_{d} \Sigma_{\theta_d \theta_d} {V}_{d}^{\top})$.
Note that the matrix ${V}_{d} \Sigma_{\theta_d \theta_d} {V}_{d}^{\top}$ is not necessarily invertible, and if not, as usual, the inverse covariance should be interpreted in the sense of pseudo-inverse.
Let ${a} =: {V}_{d}\params_d \in \BR^{d_x-1}$. Then
\begin{gather*}
    \kappa_{d} = \BE_{{a} \sim \mathcal{N}(a; {0}, {V}_{d} \Sigma_{\theta_d \theta_d} {V}_{d}^{\top})}\Big[\sum_{i} |a_{i}|\Big] = \sum_{i} \BE_{a_{i} \sim \mathcal{N}(a_i; {0}, {v}_{i} \Sigma_{\theta_d \theta_d} {v}_{i}^{\top})}[ |a_{i}|].
\end{gather*}
The distribution of $|a_{i}|$ follows a half-normal distribution, and there holds (cf.~eq.~(3) of \cite{LeoneNelson:1961})
\begin{equation*}
    \mathbb{E}_{a_i\sim \mathcal{N}({0}, {v}_{i} \Sigma_{\theta_d \theta_d} {v}_{i}^{\top})}[|a_i|]=\sqrt{\frac{2}{\pi}}({v}_{i} \Sigma_{\theta_d \theta_d} {v}_{i}^{\top})^\frac12.
\end{equation*}
Consequently,
\begin{equation}\label{eq:kappa_d-exp}
    \kappa_{d} =  \sqrt{\frac{2}{\pi}}\sum_{i}({v}_{i} \Sigma_{\theta_d \theta_d} {v}_{i}^{\top})^\frac12\quad \mbox{and}\quad \frac{\partial\kappa_{d}}{\partial\ell_d} = \sqrt{\frac{1}{2\pi}}\sum_{i} ({v}_{i} \Sigma_{\theta_d \theta_d} {v}_{i}^{\top})^{-\frac12} {v}_{i} \frac{\partial}{\partial\ell_d}\Sigma_{\theta_d \theta_d} {v}_{i}^{\top}.
\end{equation}
It remains to examine the monotonicity of ${v}_{i} \Sigma_{\theta_d \theta_d} {v}_{i}^{\top}$ in $\ell_{d}$. Indeed, by the definition of $\Sigma_d$, direct computation gives
\begin{equation*}
\frac{\partial}{\partial\ell_d} [\Sigma_{\theta_d \theta_d}(\ell_d)]_{j,j'}=\frac{\partial}{\partial\ell_d}\sigma_d^2\exp\Big(-\frac{\dist(j,j')}{\ell_d}\Big)=\frac{\sigma_d^2\dist(j,j')}{\ell_d^2}\exp\Big(-\frac{\dist(j,j')}{\ell_d}\Big),
\end{equation*}
and thus 
\begin{equation*}
    \frac{\partial}{\partial\ell_d}{v}_{i} \Sigma_{\theta_d \theta_d} {v}_{i}^{\top}=\frac{\sigma_d^2}{\ell_d^2}\sum_{j}\sum_{j'}{v}_{i,j}\dist(j,j')\exp\Big(-\frac{\dist(j,j')}{\ell_d}\Big) {v}_{i,j'}.
\end{equation*}
Then it follows that if the vectors ${v}_i$ were arbitrary, the monotonicity issue would rest on the positive definiteness of the associated derivative kernel. For example, for a Gaussian kernel $e^{-\frac{(x-y)^2}{\ell_d}}$ (i.e. $\dist$ is the squared Euclidean distance), the associated kernel $k(x,y)$ is given by $(x-y)^2e^{-\frac{(x-y)^2}{\ell_d}}$. This issue seems generally challenging to verify directly, since $(x-y)^2$ is not a positive semidefinite kernel by itself on $\mathbb{R}$, even though the Gaussian kernel $e^{-\frac{(x-y)^2}{\ell_d}}$ is indeed positive semidefinite. Thus, one cannot use the standard Schur product theorem to conclude the monotonicity. Alternatively, one can also compute the Fourier transform of the kernel $k(x)=x^2e^{-x^2}$ directly, which is given by 
\begin{equation*}
    \mathcal{F}[k(x)](\omega)=\frac{2-\omega^2}{4}\frac{1}{\sqrt{2}}e^{-\frac{\omega^2}{4}}.
\end{equation*}
see the proposition below for the detailed derivation.
Clearly, the Fourier transform of the kernel $x^2e^{-x^2}$ is not positive over the whole real line $\mathbb{R}$. By Bochner's theorem (see e.g. p. 19 of \cite{Rudin:1990}), this kernel is actually not positive.
The fact that the kernel is no longer positive definite makes the analytical analysis challenging.  This observation holds also for the Matern-$\nicefrac{1}{2}$ kernel, see the proposition below. These observations clearly indicate the risk for a potential non-monotonicity in $\ell$. Nonetheless, we emphasise that this condition is only sufficient, but not necessary, since the kernel is only evaluated at lattice points (instead of arbitrary scattered points). We leave a full investigation of the monotonicity to a future work, given the compelling empirical evidence for monotonicity in both the NN and linearised settings. 

The next result collects the Fourier transforms of the associated kernel for the Gaussian and Matern-$\nicefrac{1}{2}$ kernels.
\begin{proposition}
The Fourier transforms of the functions $x^2e^{-x^2}$ and $|x|e^{-|x|}$ are given by
\begin{equation*}
    \mathcal{F}[x^2e^{-x^2}](\omega) = \frac{2-\omega^2}{4\sqrt{2}}e^{-\frac{x^2}{4}}\quad \mbox{and}\quad 
    \mathcal{F}[|x|e^{-|x|}](\omega)=\frac{2(1-\omega^2)}{\sqrt{2\pi}(1+\omega^2)^2}.
\end{equation*}
\end{proposition}
\begin{proof}
Recall that the Fourier transform $\mathcal{F}[e^{-x^2}]$ of the Gaussian kernel $e^{-x^2}$ is given by 
\begin{equation*}
    \mathcal{F}[e^{-x^2}](\omega) = \frac{1}{\sqrt{2\pi}}\int_{-\infty}^\infty e^{-x^2}e^{-{\rm i}\omega x}{\rm d}x = \frac{1}{\sqrt{2}}e^{-\frac{\omega^2}{4}}.
\end{equation*}
Direct computation shows 
\begin{equation*}
     k''(x)=4x^2e^{-x^2}-2e^{-x^2} = 4x^2e^{-x^2}-2k(x). 
\end{equation*}
Taking Fourier transform on both sides and using the identity $\mathcal{F}[k''(x)](\omega)=-\omega^2\mathcal{F}[f(x)](\omega)$, we obtain 
\begin{equation*}
    -\omega^2\mathcal{F}[f(x)](\omega) = 4 \mathcal{F}[x^2e^{-x^2}](\omega)-2\mathcal{F}[f(x)](\omega),
\end{equation*}
which upon rearrangement gives the desired expression for  the Fourier transform $\mathcal{F}[x^2f(x)]$. Next we compute $\mathcal{F}[|x|e^{-|x|}](\omega)$:
\begin{align*}
     &\quad \mathcal{F}[|x|e^{-|x|}](\omega)=\frac{1}{\sqrt{2\pi}}\int_{-\infty}^\infty |x|e^{-|x|}e^{-{\rm i}\omega x}{\rm d}x \\ &=\frac{1}{\sqrt{2\pi}}\int_{-\infty}^\infty |x|e^{-|x|}(\cos \omega x -{\rm i}\sin \omega x){\rm d}x
    = \frac{2}{\sqrt{2\pi}}\int_0^\infty xe^{-x}\cos \omega x{\rm d}x,
\end{align*}
since the term involving $\sin\omega x$ is odd and the corresponding integral vanishes. Integration by parts twice gives 
\begin{align*}
    \int_0^\infty xe^{-x}\cos \omega x{\rm d}x &= -xe^{-x}\cos \omega x|_{x=0}^\infty + \int_0^\infty e^{-x} (\cos \omega x-\omega x\sin \omega x){\rm d}x\\ 
    & =  \int_0^\infty e^{-x} \cos \omega x{\rm d}x -\int_0^\infty\omega xe^{-x}\sin \omega x{\rm d}x\\
    & = \int_0^\infty e^{-x} \cos \omega x{\rm d}x + \omega x e^{-x}\sin \omega x|_{x=0}^\infty - \int_0^\infty e^{-x}(\omega \sin\omega x+\omega^2x\cos\omega x){\rm d}x.
\end{align*}
Rearranging the identity gives
\begin{equation*}
    \int_0^\infty xe^{-x}\cos \omega x{\rm d}x = \frac{1}{\omega^2+1}\int_0^\infty e^{-x}\cos \omega x {\rm d}x- \frac{\omega}{\omega^2+1}\int_0^\infty e^{-x}\sin \omega x{\rm d}x
\end{equation*}
This and the identities
\begin{align*}
    \int_0^\infty e^{-x}\cos \omega x {\rm d}x = \frac{1}{1+\omega^2} \quad\mbox{and}\quad 
    \int_0^\infty e^{-x}\sin \omega x {\rm d}x = \frac{\omega}{1+\omega^2},
\end{align*}
immediately imply 
\begin{equation*}
    \mathcal{F}[|x|e^{-|x|}](\omega)=\frac{2}{\sqrt{2\pi}}\int_0^\infty xe^{-x}\cos \omega x{\rm d}x = \frac{2(1-\omega^2)}{\sqrt{2\pi}(1+\omega^2)^2}.
\end{equation*}
This shows the second identity.
\end{proof}

\section{Additional experimental discussion}\label{apd:add_exp}

In this section, we provide additional empirical evaluation of the uncertainty estimates obtained with the linearised DIP.
Validating the accuracy of the uncertainty estimates is crucial for their reliable integration into downstream tasks and computer human interaction workflows, as discussed by \cite{antoran2021getting}, \cite{Bhatt2021uncertainty}, and \cite{BarbanoArridgeJinTanno:2021}.

\subsection{Evaluating approximate computations}

We validate the accuracy of our approximate computation presented in \cref{sec:computations} on the KMNIST dataset.
KMNIST is the perfect ground for this evaluation due to the fact that the low-dimensionality of $d_x$ and $d_y$ guarantees computational tractability of the inference problem, allowing us to benchmark the approximations we introduce in \cref{sec:computations}, against exact computation. In this section, if not stated otherwise, we carry out our investigations with the setting where the forward operator $\op$, comprises 20 angles, and we add $5\%$ noise to $\op x$. We repeat the analysis on 10 characters taken from the test set of the KMNIST dataset.
We assess the suitability of the Hutchinson trace estimator for the gradient of the log-determinant (\cref{subsec:cg_logdet_grads}), and the ancestral sampling for the TV-PredCP gradients (\cref{sec:predcptv_computation}).
\Cref{fig:kmnist_hyperparams_approx_vs_exact_include_predcp_False} and \cref{fig:kmnist_hyperparams_approx_vs_exact_include_predcp_True} show hyperparameter optimisation $(\sigma_y^2, \sigma^2, \ell)$ using exact and estimated gradients.
The hyperparameters trajectories match closely; we only observe tiny oscillations when using estimated gradients.
The log-determinant gradients $\frac{\partial \text{log}|\Sigma_{yy}|}{\partial \phi}$ are estimated using 10 samples, $v\sim \mathcal{N}(v; 0, P)$.
The PCG for solving ${v}^{\top} \Sigma_{\data \data}^{-1}$ uses a maximum of 50 iterations (with a early stopping criterion in place if a tolerance of $1.0$ is met). We use a randomised SVD-based preconditioner $P$ (cf. \ref{subsec:cg_logdet_grads}), where the rank, $r$, is chosen to be 200, and $P$ is updated every 100 steps.
The TV-PredCP gradients are estimated using 500 samples.

\begin{figure*}
    \centering
    \includegraphics[width=\textwidth]{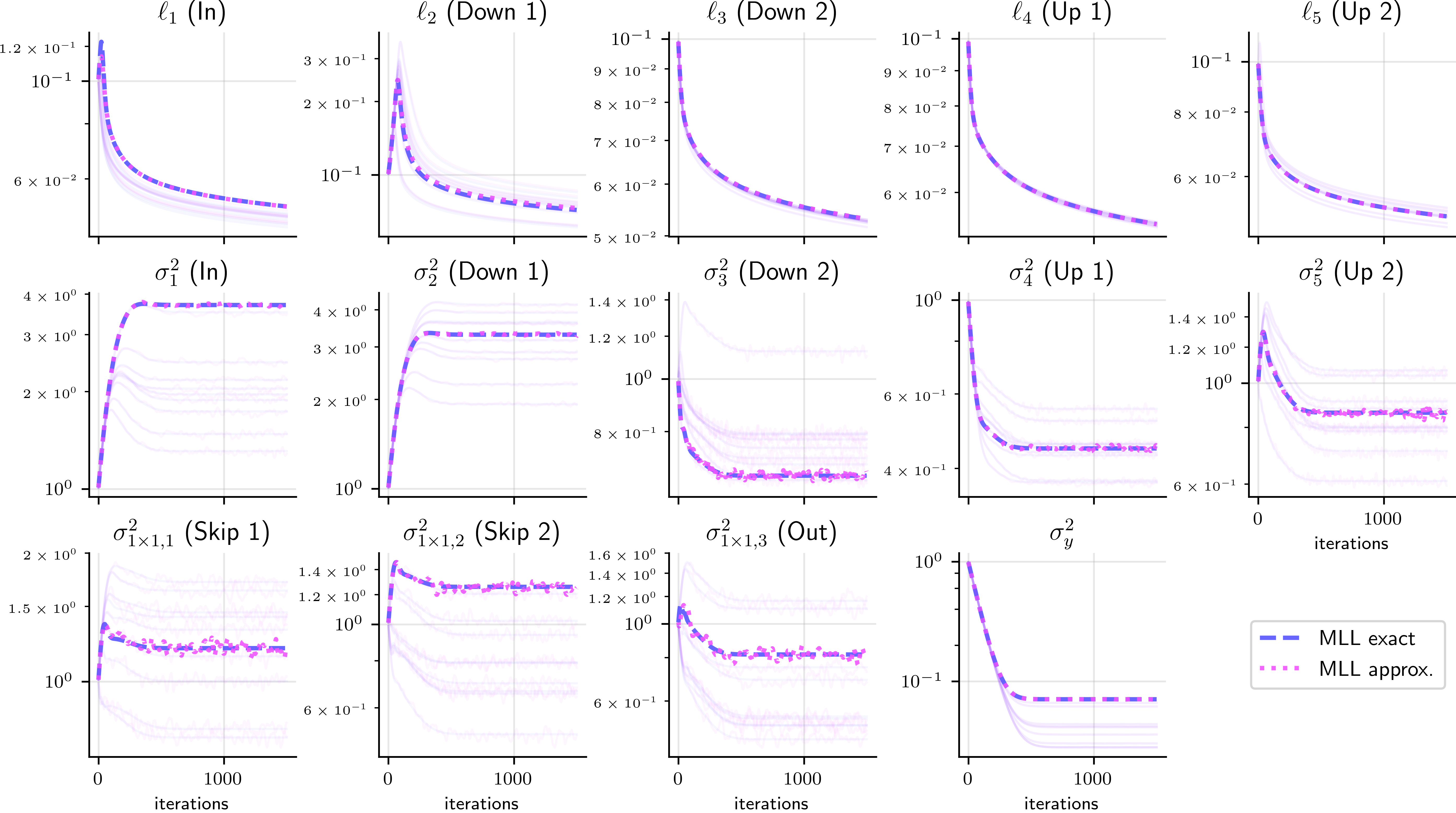}
    \caption{Hyperparameters' optimisation in \cref{eq:type2MAP} for lin.-DIP excluding PredCP (MLL), computing exact gradients as well as resorting to the approximate numerical methods discussed in \cref{subsec:cg_logdet_grads} (i.e. PCG-based log-determinant gradients) on 10 KMNIST images.}
    \label{fig:kmnist_hyperparams_approx_vs_exact_include_predcp_False}
\end{figure*}

\begin{figure*}
    \centering
    \includegraphics[width=\textwidth]{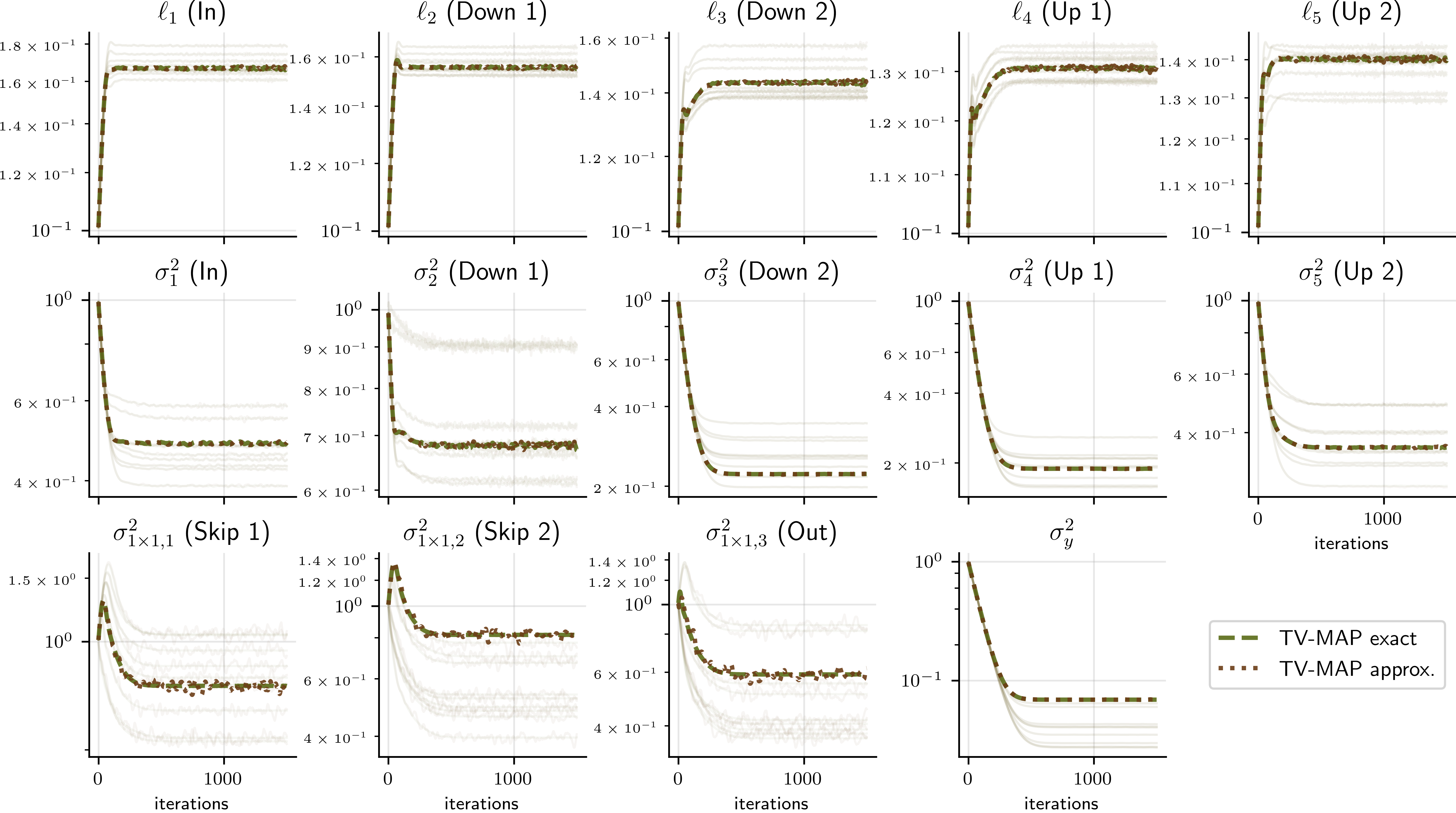}
    \caption{Hyperparameters' optimisation in \cref{eq:type2MAP} for lin.-DIP including TV-PredCP (TV-MAP), computing exact gradients as well as resorting to the approximate numerical methods discussed in \cref{subsec:cg_logdet_grads} (i.e. PCG-based log-determinant gradients) and \cref{sec:predcptv_computation} (i.e. ancestral sampling for TV-PredCP term) on 10 KMNIST images.}
    \label{fig:kmnist_hyperparams_approx_vs_exact_include_predcp_True}
\end{figure*}

\begin{figure*}
    \centering
    \includegraphics[width=0.6\textwidth]{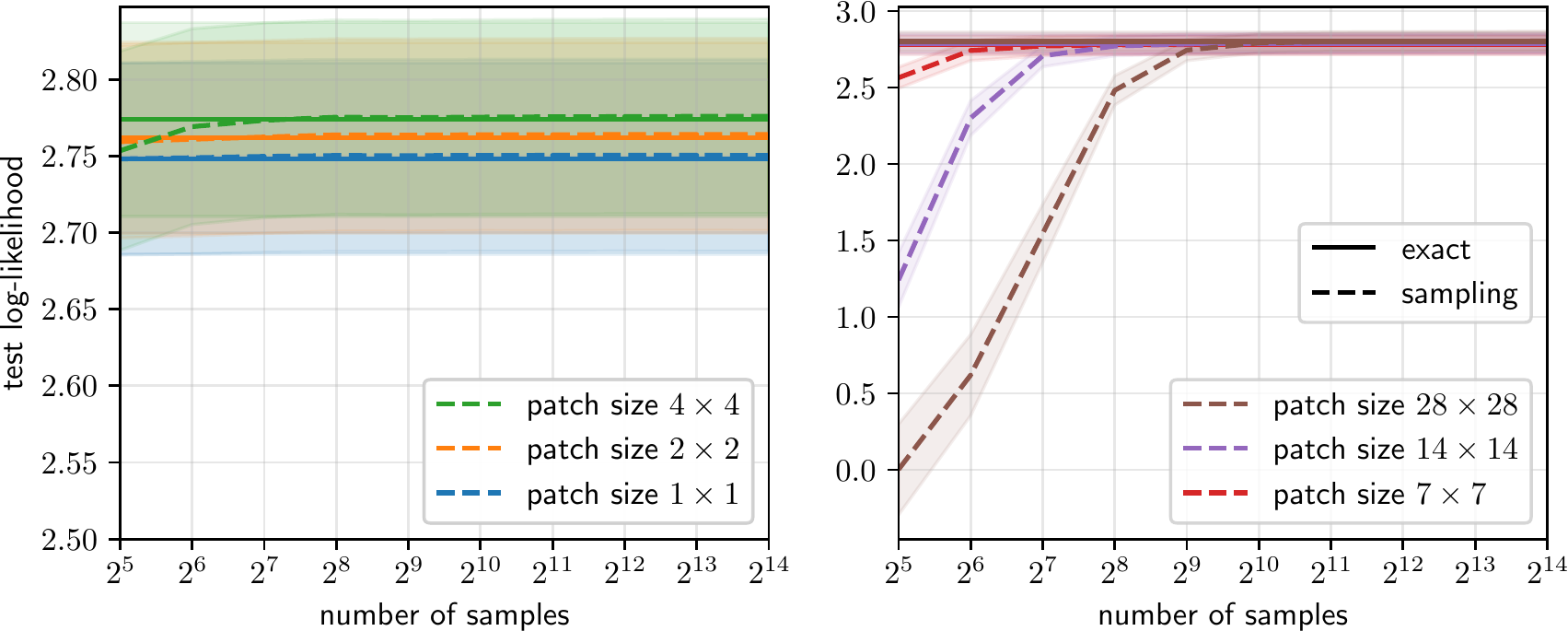}
    \caption{Test log-likelihood computed with posterior predictive covariance matrices estimated via \cref{eq:matheron}, and compared to the one obtained with exact methods (i.e. using exact posterior predictive covariance matrices via \cref{eq:CovFunc}).
    The log-likelihood is overall well approximated. As we would expect, we observe that larger patches require more samples.
    We conduct our investigation over 10 KMNIST images, and show mean and standard error. $(\sigma^2_{y}, \sigma^2, \ell)$ are obtained using MLL.}
    \label{fig:kmnist_sample_based_vs_exact}
\end{figure*}

We assess the approximations introduced in \cref{sec:post_sampling}; the accuracy of the estimation of the posterior covariance matrix, but most importantly, the estimation of the test log-likelihood. For large image sizes (e.g. the Walnut cf. \cref{subsec:walnut-exps}), it is infeasible to store the posterior predictive covariance matrix $\Sigma_{x|y} \in \mathbb{R}^{d_x\times d_x}$, which in single precision would require 250 GB of memory. However, it can be made computationally cheaper if we consider smaller image patches of pixels, neglecting the inter-patch-dependencies. This assumes the covariance matrix $\Sigma_{x|y}$ to be block diagonal. \Cref{fig:kmnist_sample_based_vs_exact} shows the effect of neglecting inter-patch-dependencies.
The log-likelihood increases with increasing patch-size (i.e.\ with more inter-dependencies being taken into account). \Cref{fig:kmnist_sample_based_vs_exact} shows how well the test log-likelihood is approximated when resorting to posterior predictive covariance matrices estimated via sampling using \cref{eq:matheron}, while sweeping across different numbers of samples and patch-sizes. As expected, estimating the log-likelihood for larger patch-sizes requires more samples.
On KMNIST, 1024 samples are sufficient for almost perfect approximation of the test log-likelihood, when approximating the posterior predictive covariance matrix with patch-size of $28\times 28$. Note that a patch-size of $28\times 28$ on KMNIST implies that no inter-patch-dependencies are neglected.

\subsection{Further discussion on KMNIST}\label{apd:KMNIST_additional_exp}

We include additional experimental figures to support the discussion about the experiments in \cref{sec:calibration_comparison}. \Cref{fig:kmnist_angles_10_noise_5_1}, \cref{fig:kmnist_angles_20_noise_5_1}, \cref{fig:kmnist_angles_10_noise_1_1}, and \cref{fig:kmnist_angles_20_noise_1_1}  are analogous to \cref{fig:main_kmnist}, yet show a KMNIST character for four different problem settings: 10 angles and 20 angles, and the two noise regimes. 

\begin{figure*}
    \centering
    \includegraphics[width=\textwidth]{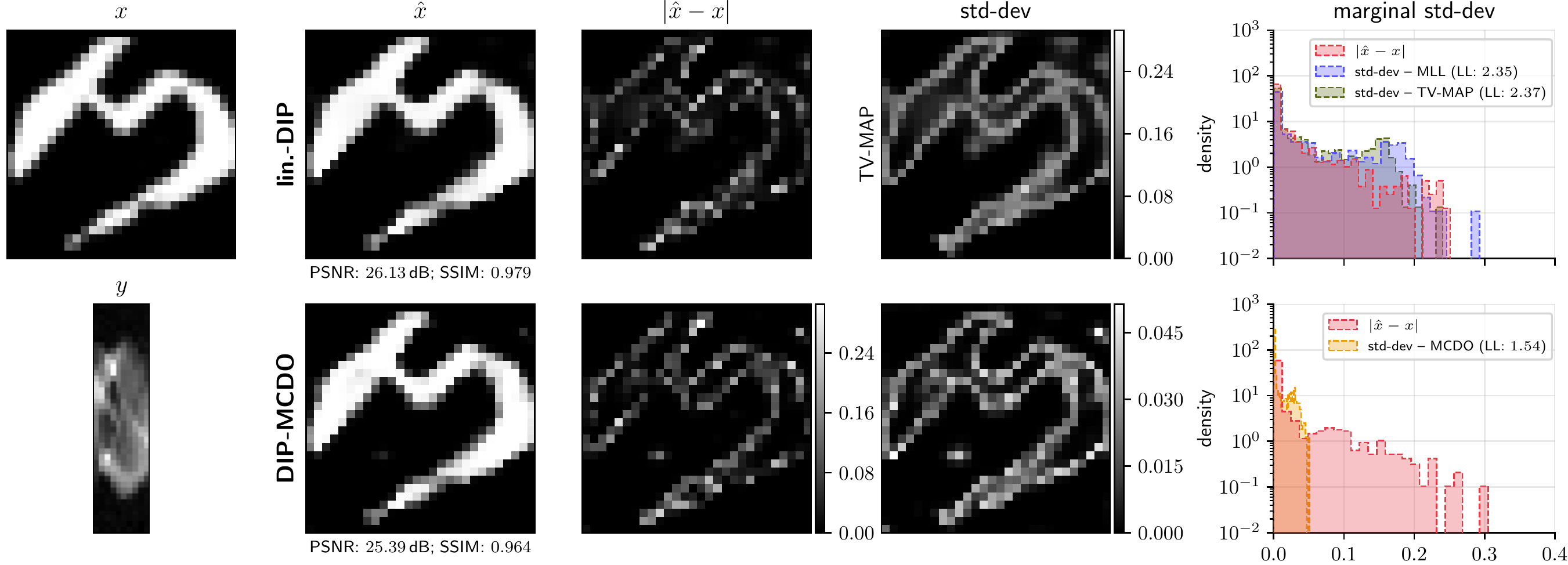}
    \caption{KMNIST character recovered from a simulated observation $\data$ (using $10$ angles and $\eta(5\%)$) with lin.-DIP, DIP-MCDO and along with their uncertainty estimates and histogram plots.}
    \label{fig:kmnist_angles_10_noise_5_1}
\end{figure*}

\begin{figure*}
    \centering
    \includegraphics[width=\textwidth]{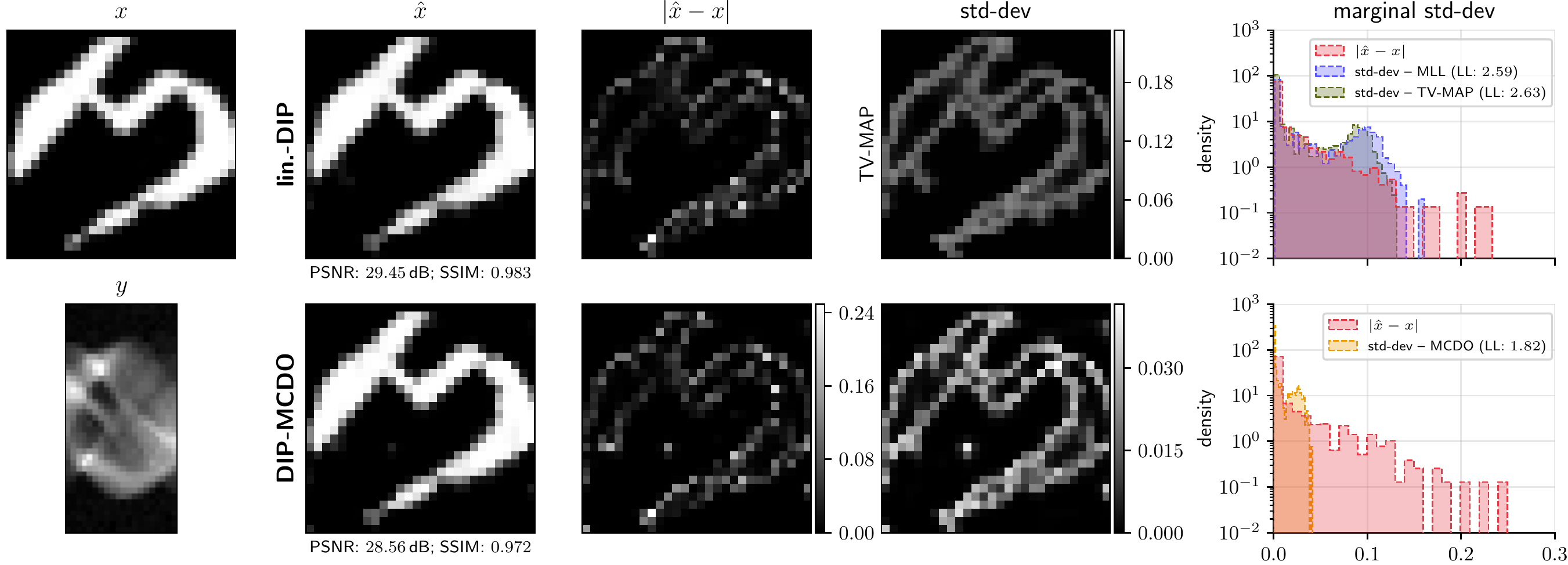}
    \caption{KMNIST character recovered from a simulated observation $\data$ (using $20$ angles and $\eta(5\%)$) with lin.-DIP, DIP-MCDO along with their uncertainty estimates and histogram plots.}
    \label{fig:kmnist_angles_20_noise_5_1}
\end{figure*}

\begin{figure*}
    \centering
    \includegraphics[width=\textwidth]{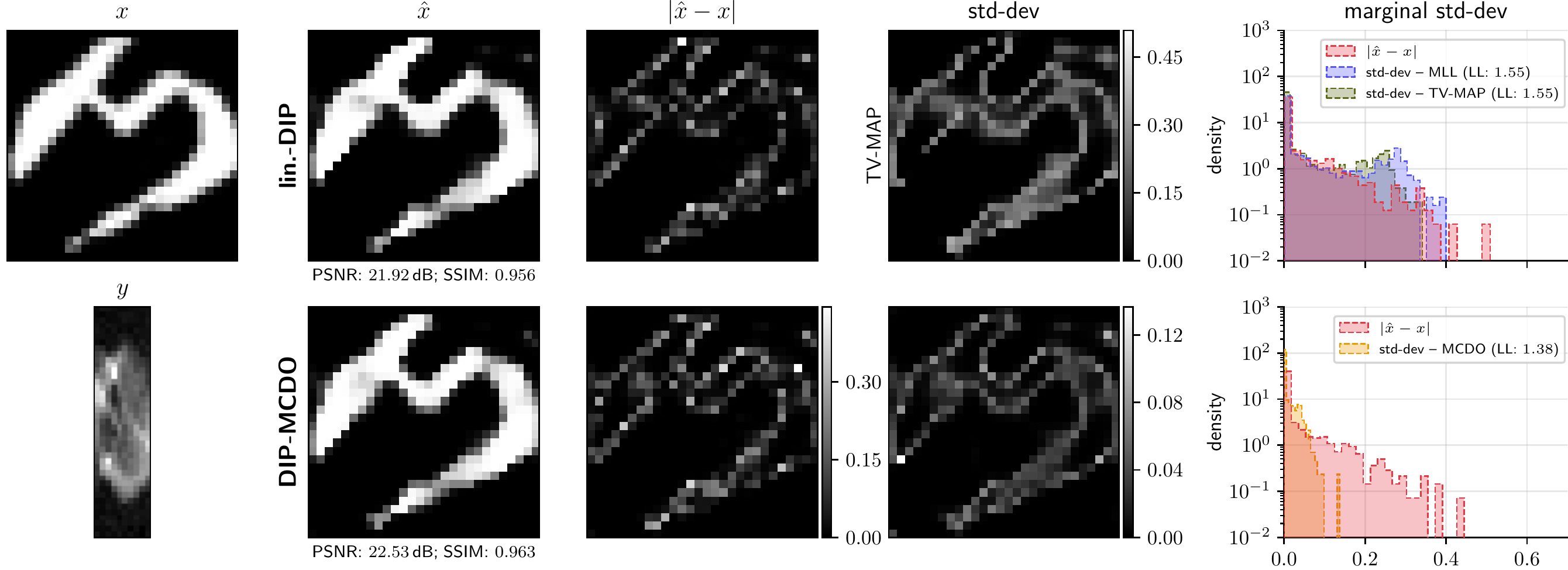}
    \caption{KMNIST character recovered from a simulated observation $\data$ (using $10$ angles and $\eta(10\%)$) with lin.-DIP, DIP-MCDO along with their uncertainty estimates and histogram plots.}
    \label{fig:kmnist_angles_10_noise_1_1}
\end{figure*}

\begin{figure*}
    \centering
    \includegraphics[width=\textwidth]{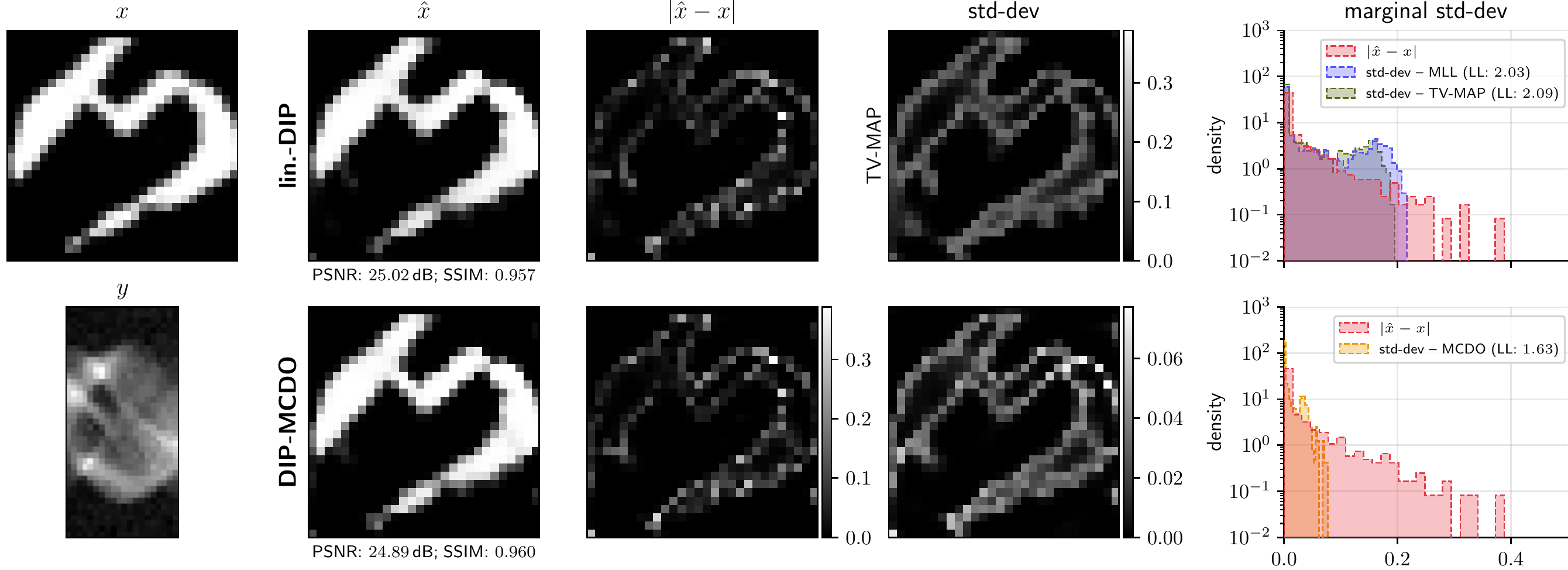}
    \caption{KMNIST character recovered from a simulated observation $\data$ (using $20$ angles and $\eta(10\%)$) with lin.-DIP, DIP-MCDO along with their uncertainty estimates and calibration plots.}
    \label{fig:kmnist_angles_20_noise_1_1}
\end{figure*}

\begin{figure*}
    \centering
    \includegraphics[width=\textwidth]{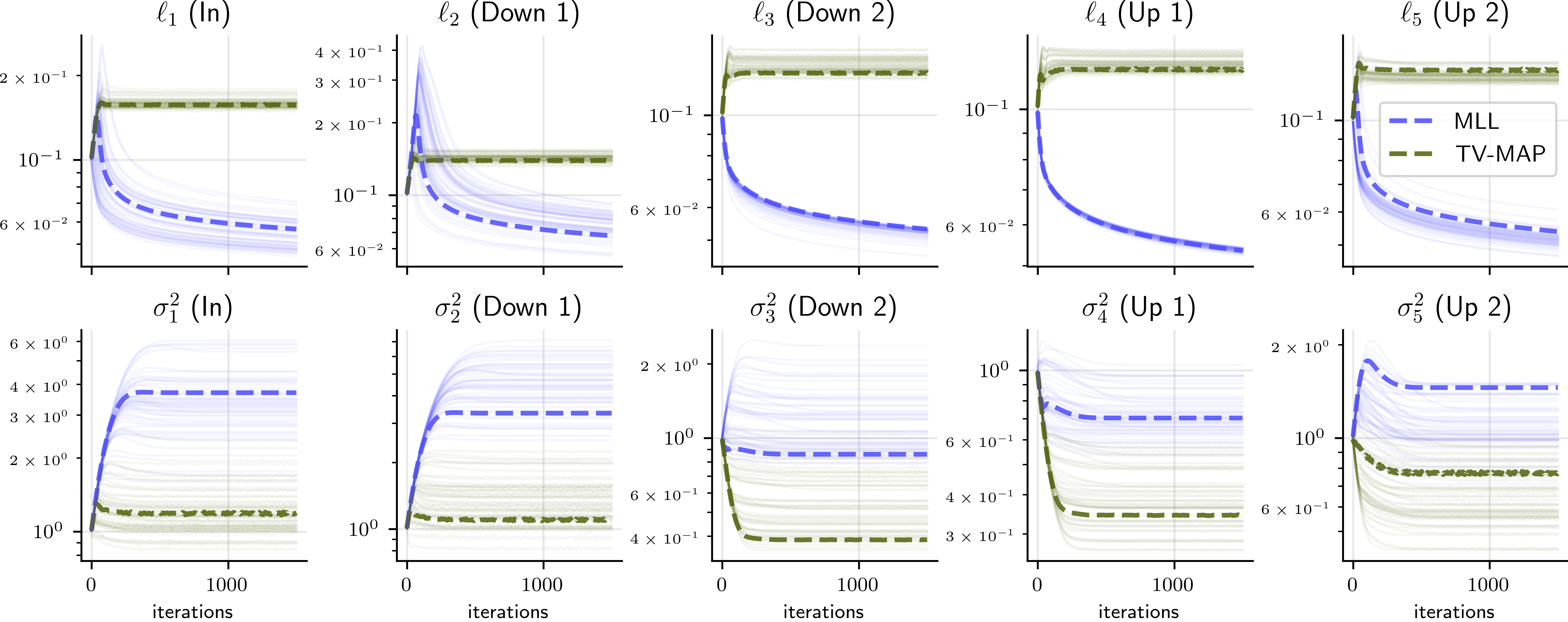}
    \caption{
    Optimisation of ($\ell, \sigma^{2}$) via MLL and Type-II MAP for $3\times 3$ convolution layers belonging to the small U-Net used for KMNIST. 
    Thicker dotted lines refer to the optimisation of the exemplary reconstruction shown in \cref{fig:main_kmnist} while transparent lines correspond to other KMNIST images.The TV-PredCP leads to larger prior lengthscales $\ell$ and lower variances $\sigma^{2}$.
    }
    \label{fig:KMNIST_hyperparams_gp_priors}
\end{figure*}

\Cref{fig:KMNIST_hyperparams_gp_priors} and \cref{fig:KMNIST_hyperparams_normal_priors} show the hyperparameters' optimisation via Type-II MAP and MLL outlined in \cref{subsec:marginal_likelihood}. 
The use of our TV-PredCP prior leads to smaller marginal variances and larger lengthscales. 
This restricts our prior over reconstructions to smooth functions.
The TV-PredCP introduces additional constraints into the model by encouraging the prior to contract (stronger parameter correlations and smaller posterior predictive marginal variances. In turn, this results in a more contracted posterior, which we observe as a larger Hessian determinant.

\begin{figure*}
    \centering
    \includegraphics[width=0.85\textwidth]{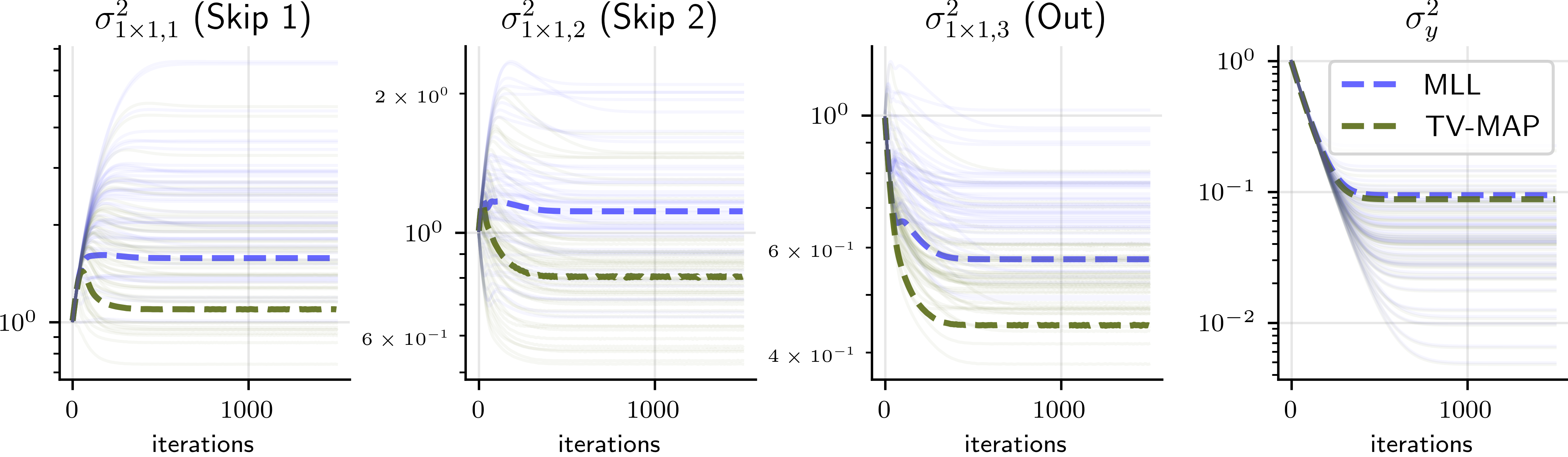}
    \caption{
    Hyperparameters' optimisation via MLL and Type-II MAP for $1\times1$ convolution layers belonging to the small U-Net used for KMNIST, along with $\sigma_y^{2}$. 
    Thicker dotted lines refer to the optimisation of the KMNIST image shown in \cref{fig:main_kmnist}, while transparent lines correspond to other KMNIST images.}
    \label{fig:KMNIST_hyperparams_normal_priors}
\end{figure*}

For the KMIST dataset, one may question whether TV is an ideal regulariser. The TV regulariser enforces sparsity in the local image gradients. A TV regulariser is highly recommended when we observe sparsity in the edges present in an image, especially when the edges constitute a small fraction of the overall image pixels. That is often the case in high-resolution medical images or natural images.
Intuitively, the higher the resolution of the image is, the higher the sparsity level of the edges is. However, in the KMIST dataset, due to the low resolution of the images, the edges constitute a considerable fraction of the total pixels.
Therefore, a TV regulariser could be sub-optimal. In the KMNIST dataset, it is difficult to distinguish (in TV sense) what is part of the image structure and what is part of the background. The stroke is only a few pixels wide, and ground-truth pixel values are generated through interpolation \cite{clanuwat2018deep}. Indeed we observe a larger gain from selecting hyperparamerters using Type2-MAP (instead of MLL) for the real-measured high-resolution Walnut data than for KMNIST.

Furthermore, some KMNIST images present spurious high valued pixels away from the region containing the handwritten character. This contradicts the modelling assumption in \cref{eq:inverse_problem} which assumes $x$ is noiseless. Our likelihood function from \cref{eq:Model_weight_space} is defined over the space of observations $\data$ and thus can not account for noise in $x$.
We translate the uncertainty induced by the observation noise to the space of images by computing the conditional log-likelihood Hessian with respect to $x$: $-\frac{\partial^2 \log p(\data|x)}{\partial x^2} =\sigma^{-2}_{y}\op^{\top}\op \in \BR^{d_{x} \times d_{x}}$. This matrix is of rank at most $d_{y}$, which potentially can be much smaller than $d_{x}$ due to the ill-conditioning of the reconstruction problem, and therefore cannot act as a proper Gaussian precision matrix on its own. We incorporate the noise uncertainty from the observation subspace into the image space by adding the mean of the diagonal of the pseudoinverse $ \sigma^{2}_{y}(\op^{\top}\op)^{\dagger}$ to the marginal variances of the predictive distribution.
This can also be seen as placing a Gaussian likelihood over reconstruction space, which can be marginalised to recover the predictive distribution $p(x|\data) = \int \mathcal{N}(x; \hat x, \sigma^{2}_{y}\text{Tr}((\op^{\top}\op)^{\dagger}) d^{-1}_{x} \mI) \mathcal{N}(\params; \hat\params, \Sigma_{\params|\data})\,{\rm d}\params = \mathcal{N}(x; \hat x,  \mJ \Sigma_{\params|\data} \mJ^{\top} + \sigma^{2}_{y}\text{Tr}((\op^{\top}\op)^{\dagger}) d^{-1}_{x} \mI)$.

\subsection{Further discussions on $\mu$CT Walnut data}\label{apd:walnut_additional_exp}

We include additional figures to support the discussion about the experiments in \cref{subsec:walnut-exps}. 
We evaluate the effect of the TV-PredCP prior for hyperparameter optimisation. We observe that this prior leads to a slightly less heavy tailed standard deviation histogram. It presents slightly better agreement with the empirical reconstruction error, resulting in a lager log-likelihood.
\Cref{fig:walnut_hyperparams_gp_priors} and \cref{fig:walnut_hyperparams_normal_priors} show the optimisation of the hyperparameters ($\sigma_y^{2}$, $\ell$, $\bsigma^{2}_{\params}$) using the method in \cref{subsec:marginal_likelihood} and approximate computations in \cref{sec:computations}. For both MLL and Type2-MAP learning, the marginal variance for all CNN blocks except the two closest to the output goes to $\approx$0. 
This is due to the representations from these last layer being able to explain the data well on their own. 
The our hyperparameter objectives are thus able to eliminate previous layers from our probabilistic model, simplifying it without sacrificing reconstruction quality.  We did not observe this for KMNIST data, possibly because of our use of a smaller, less overparametrised network without any spare capacity.

\begin{figure*}
    \centering
    \includegraphics[width=\textwidth]{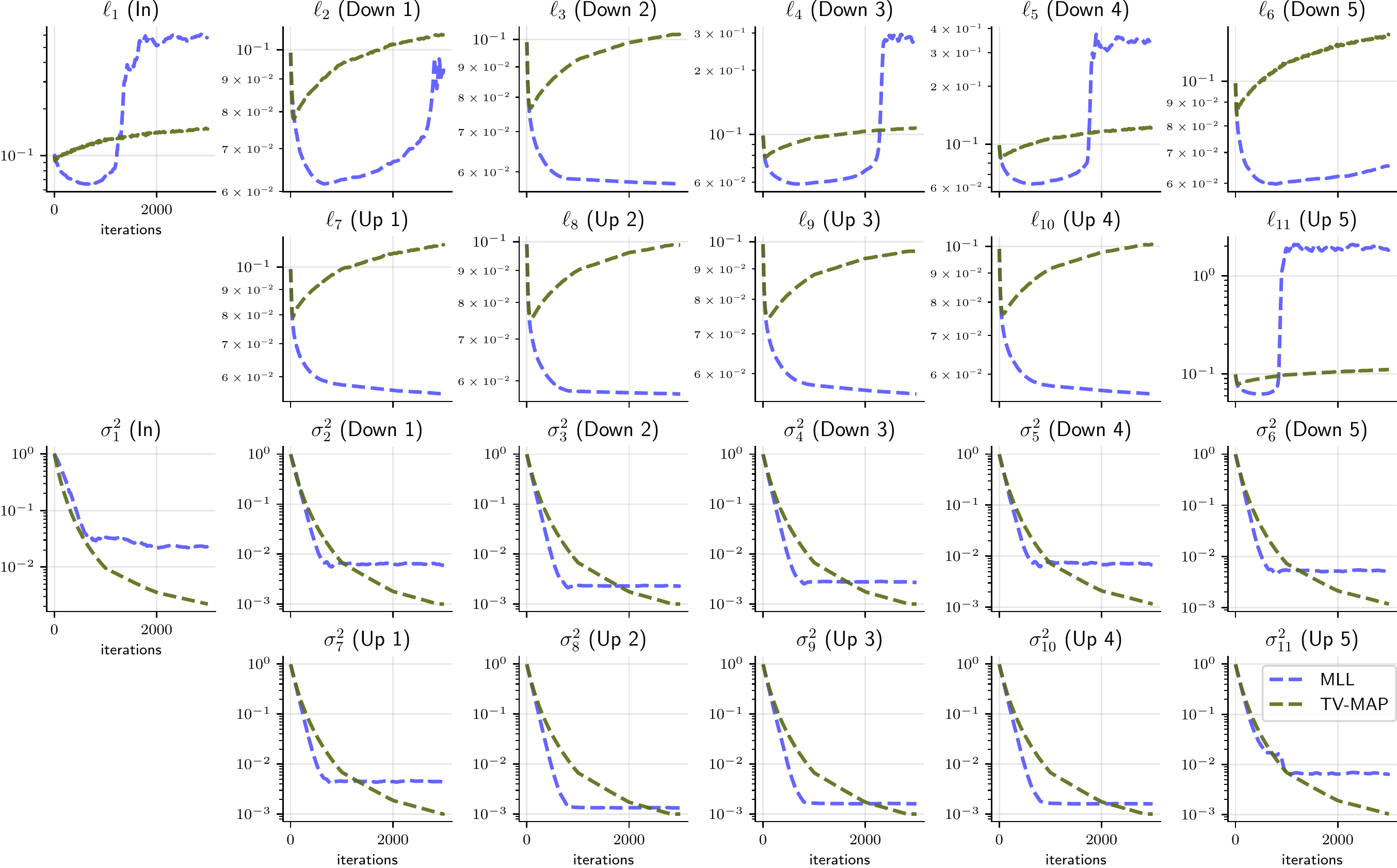}
    \caption{Optimisation of ($\ell, \sigma^{2}$) via MLL and Type-II MAP for $3 \times 3$ convolution layers for the Walnut data described in \cref{subsec:walnut-exps}. As in \cref{fig:KMNIST_hyperparams_gp_priors}, the TV-PredCP leads to larger prior lengthscales $\ell$ and lower variances $\sigma^{2}$.}
    \label{fig:walnut_hyperparams_gp_priors}
\end{figure*}

\begin{figure*}
    \centering
    \includegraphics[width=0.75\textwidth]{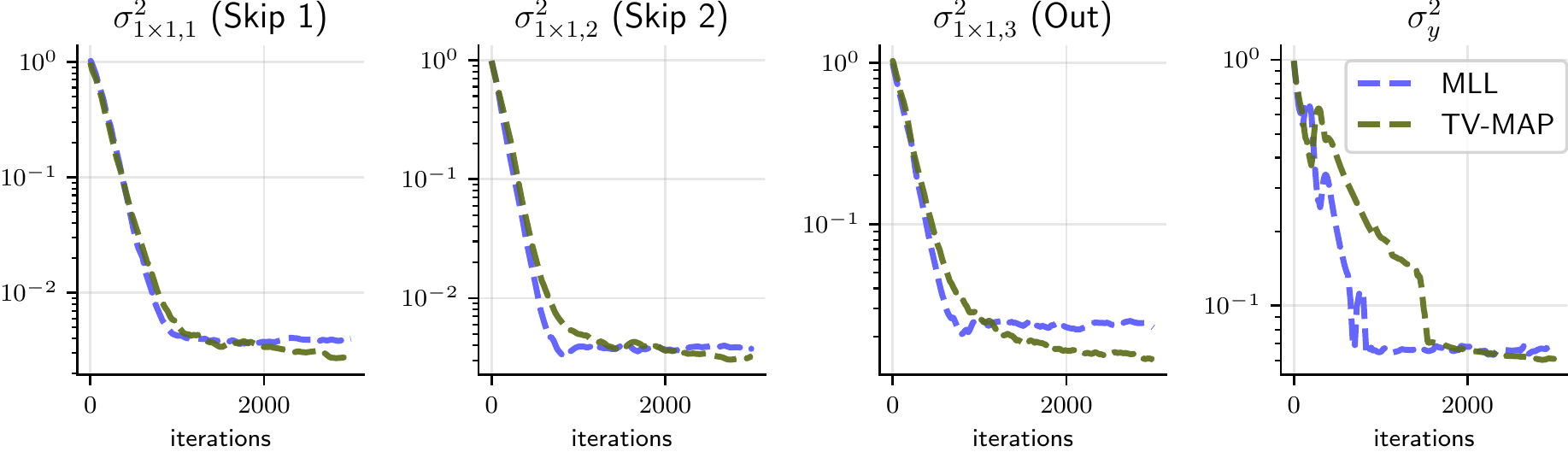}
    \caption{Optimisation of ($\sigma_y^{2}$, $\sigma^{2}$) via MLL and Type-II MAP for $1 \times 1$ convolutions.}
    \label{fig:walnut_hyperparams_normal_priors}
\end{figure*}

\section{Additional experimental setup details}\label{apd:exp_setup}

\subsection{Setup for KMNIST experiments}\label{apd:add_setup_kmnist}

We use a down-sized version of U-Net \cite{ronneberger2015u}, cf. \cref{fig:unet-diagram-kmnist}, as the reduced output dimension $d_{x}$ and the simplicity of the problem allow us to employ a shallow architecture without compromising the reconstruction quality. This problem is computationally tractable removing the need for the approximations described in \cref{sec:computations}.
We reduce the U-Net architecture in \cref{fig:unet-diagram} to 3 scales and 32 channels at each scale, remove group-normalisation layers and use a sigmoid activation for the output. A filtered back-projection reconstruction from $\data$ is used as the network input.

\begin{figure}[h]
    \centering
    \includegraphics[width=0.6\columnwidth]{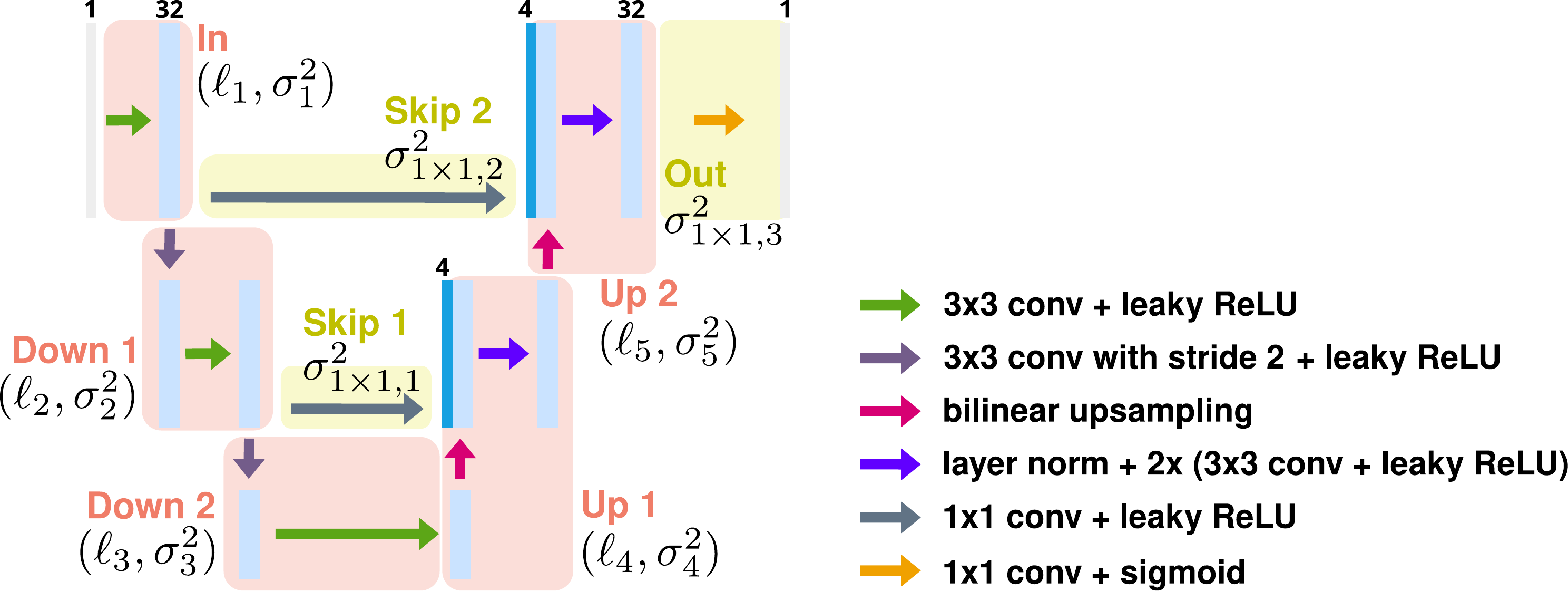}
    \vspace{-0.5em}
    \caption{A schematic illustration of the reduced U-Net architecture used in the numerical experiments on KMNIST data. It has 3 scales and does not include group norm layers. Each light-blue rectangle corresponds to a multi-channel feature map. We highlight the architectural components corresponding to each block for which a separate prior is defined with red boxes.}
    \vspace{-0.75em}
    \label{fig:unet-diagram-kmnist}
\end{figure}

\Cref{tab:hparams_kmnist} lists the hyperparameters of DIP optimisation for each setting. These values were found by grid-search on 50 KMNIST training images.
The dropout rate $p$ of DIP-MCDO is set to $0.05$.

\begin{table}[H]
\centering
\caption{Hyperparameters of DIP optimisation selected using 50 randomly chosen images from the KMNIST training set. The $\lambda$ values refer to our implementation of \cref{eq:DIP_MAP-obj2} in which $\Vert\cdot\Vert^2$ is replaced with mean squared error (or the regularisation term is up-scaled by $d_{x}$).}
\label{tab:hparams_kmnist}
\begin{tabular}{rrrrrrcrrrr}
 & & \multicolumn{4}{c}{5\% noise} & & \multicolumn{4}{c}{10\% noise}\\
$\#$angles & & 5 & 10 & 20 & 30 & & 5 & 10 & 20 & 30\\
\hline
TV scaling for DIP: $\lambda$ & & $\rm 1e{-}5$ & $\rm 3e{-}5$ & $\rm 1e{-}4$& $\rm 1e{-}4$ & & $\rm 3e{-}5$ & $\rm 1e{-}4$ & $\rm 3e{-}4$ & $\rm 3e{-}4$\\
DIP iterations & & 14\,000 & 29\,000 & 41\,000 & 50\,000 & & 7400 & 13\,000 & 17\,000 & 22\,000\\
\hline
\end{tabular}%
\end{table}

\subsection{Setup for X-ray Walnut data experiments}\label{apd:add_setup_walnut}

In \cite{der_sarkissian2019walnuts_data} projection data sets obtained with three different source positions are provided for 42 walnuts, as well as high-quality reconstructions of size $501^3$\,px obtained via iterative reconstruction using the measurements from all three source positions.
We consider the task of reconstructing a single slice of size $501^2$ of the first walnut from a sub-sampled set of measurements using the second source position, which corresponds to a sparse fan-beam-like geometry. From the original 1200 projections (equally distributed over $360^\circ$) of size $972\times 768$ we first select the appropriate detector row matching the slice position (which varies for different detector columns and angles due to a tilt in the setup), yielding measurement data of size $1200\cdot 768$. We then sub-sample in both angle and column dimensions by factors of $20$ and $6$, respectively, leaving $d_y = 60\cdot 128 = 7680$ measurements. For evaluation metrics, we take the corresponding slice from the provided high-quality reconstruction as the reference ground truth image $ x$. The sparse operator matrix $\op$ is assembled by calling the forward projection routine of the \texttt{ASTRA} toolbox \cite{aarle2015astra} for every standard basis vector, $\op = \op [e_{1}, e_{2},\,...\, e_{d_{x}}]$.
While especially for large data dimensions it would be favourable to directly use the matrix-free implementations from the toolbox, we also need to evaluate the transposed operation ${v}_{y}^{\top} \op$, which would be only approximately matched by the back-projection routine (especially for the tilted 2D sub-geometry, which would require padding). Therefore, we resort to the sparse matrix multiplication via PyTorch.

The network architecture is shown in \cref{fig:unet-diagram}. Following the approach in \cite{barbano2021deep}, we pretrain the network to perform post-processing of filtered back-projection (FBP) reconstructions on synthetic data. The dataset consists of pairs of images containing random ellipses, and corresponding FBPs from observations simulated according to \cref{eq:inverse_problem} with $5\%$ noise. The supervised pretraining accelerates the convergence of the subsequent unsupervised DIP phase for reconstruction from $\data$. In the DIP phase, the FBP of $\data$ is used as the network input. \Cref{tab:hparams_walnut} lists the hyperparameters of DIP optimisation.
The dropout rate $p$ of DIP-MCDO is set to $0.05$.

After DIP optimisation, following \cite{antoran2022Adapting} the network weights are refined for the linearised model (\cref{eq:linearised_model}). 
We optimise the same loss function as for DIP, but with the linear model \cref{eq:linearised_model} instead of the network model, for 1000 steps. This yields network weights that fit better the subsequent MLL / Type-II-MAP optimisation \cref{eq:type2MAP}, which employs the linear model.

\begin{table}[H]
\centering
\caption{Hyperparameters of DIP optimisation used for the walnut data. The $\lambda$ value refers to our implementation of \cref{eq:DIP_MAP-obj2} in which $\Vert\cdot\Vert^2$ is replaced with mean squared error (or the regularisation term is upscaled by $d_{x}$).}
\label{tab:hparams_walnut}
\small
\begin{tabular}{rr}
\hline
TV scaling for DIP: $\lambda$ & $\rm 6.5e{-}6$ \\
DIP iterations (after pretraining) & 1500\\
\hline
\end{tabular}
\end{table}

In MLL / Type-II-MAP optimisation \cref{eq:type2MAP}, we use 10 probes to estimate the gradients of the log-determinant $\log |\Sigma_{\data \data}|$ \cref{eq:cg_logdet_grad}, employing the PCG method for solving ${v}^{\top} \Sigma_{\data \data}^{-1}$ using a maximum of 50 steps with a randomised SVD-based preconditioner $P$ of rank 200 that is updated every 100 steps.
The TV-PredCP gradients \cref{eq:tv_grad_est,eq:tv_grad_est_2} are estimated using 20 samples. The MLL / Type-II-MAP optimisation is run for 3000 iterations.

The posterior predictive covariance matrices for all methods are estimated drawing 4096 zero-mean samples and computing empirical posterior predictive covariance matrix. The latter is done for patch-sizes from $1\times 1$ up to $10\times 10$ image patches. We use a stabilising heuristic for the estimated covariance matrices, inspired by \cite{Maddox2019Simple}: by letting $\tilde\Sigma_{x|y} \leftarrow \alpha\Sigma_{x|y} + (1-\alpha)\diag(\diag( \Sigma_{x|y}))$, $\alpha=\frac{1}{2}$, the impact of the off-diagonal entries is reduced. Note that our Gaussian assumption is correct in the case of linearised DIP but not for MCDO. However, MCDO does not provide a closed form density over the reconstructed image, only samples. The dimensionality of the reconstruction is too large for exact density estimation on real-measured data. We thus compute the log-likelihood in the same way as for the linearised DIP, i.e.\ via a Gaussian distribution with mean and posterior predictive covariance matrices estimated from samples. The accelerated sampling method via $\tilde J$ \& PCG uses a randomised SVD-based 500-rank approximation $\tilde J$ of the Jacobian, and PCG for solving ${v}^{\top} \Sigma_{\data \data}^{-1}$ with a maximum of 50 steps along with a randomised SVD-based preconditioner of rank 400. This sampling variant can be performed in single precision (32 bit floating point). Thus constructing $\tilde J$ is actually much faster than reported in \cref{tab:wall-clock} (0.5\,min instead of 0.2\,h).

\end{document}